\documentclass[draftclsnofoot,twoside,onecolumn,letter,12pt]{IEEEtran}

\usepackage{graphicx,subfigure,listings,psfrag,stfloats,color,dsfont,cite}
\usepackage{amssymb,amsmath,amsthm}
\usepackage{algorithm,algorithmic}
\usepackage{balance}
\usepackage{bbm}
\usepackage[mathscr]{euscript}

\DeclareMathAlphabet{\mathpzc}{OT1}{pzc}{m}{it}

\DeclareMathAlphabet{\mathitsf}{OML}{cmbr}{m}{it}

\newtheorem{lemma}{Lemma}

\newtheorem{proposition}{Proposition}

\DeclareMathOperator{\E}{\mathds{E}}

\newcommand{\argmax}{\mathop{\mathrm{arg\,max}}}

\newcommand{\B}[1]{\mathbf{#1}}

\newcommand{\EX}[1]{\E\left\{{#1}\right\}}
\newcommand{\EXs}[2]{\E_{{#1}}\left\{{#2}\right\}}

\newcommand{\diag}[1]{\mathsf{diag}\left(#1\right)}

\providecommand{\keywords}[1]{\textbf{\textit{Index terms---}}#1}
\newcommand{\CN}[2]{\mathscr{CN}\left({#1},{#2}\right)}
\newcommand{\CNn}[2]{\mathscr{CN}_{#1}\left({#2}\right)}

\newcommand{\eav}{\textrm{E}}

\newcommand{\snr}[1]{\textrm{snr}_{\textrm{#1}}}
\newcommand{\snrE}[1]{\textrm{snr}_{\textrm{E}#1}}

\newcommand{\ds}{\displaystyle}

\begin{document}
	
	\title{
		Cell-free Massive MIMO Networks: Optimal Power Control  against Active Eavesdropping
	}
	\author{
		\IEEEauthorblockN{Tiep M. Hoang,
			Hien Quoc Ngo,
			Trung Q. Duong, Hoang D. Tuan, and Alan Marshall
		}
	\thanks{
		This work was supported in part by the U.K. Royal Academy of Engineering Research Fellowship under Grant RF1415$\backslash$14$\backslash$22, by a U.K. Engineering and Physical Sciences Research Council under Grant EP/P019374/1, by a Research Environment Links grant, ID 339568416, under the Newton Programme Vietnam partnership, and Newton Prize 2017.
	}
	\thanks{
		T. M. Hoang, H. Q. Ngo, and T. Q. Duong are with the Queen's University of Belfast, Belfast BT7 1NN, the United Kingdom (e-mail: \{mhoang02, hien.ngo, trung.q.duong\}@qub.ac.uk).
	}
	\thanks{
		H. D. Tuan is with the University of Technology Sydney, Ultimo, NSW, Australia (e-mail: tuan.hoang@uts.edu.au).
	}
	\thanks{
		A Marshall is with the University of Liverpool, Liverpool L69 3GJ, the United Kingdom (e-mail: alan.marshall@liverpool.ac.uk).
	}
}
\markboth
	{
	}{
	}
	
	\maketitle
	
	\begin{abstract}
		This paper studies the security aspect of  a recently introduced network (``cell-free massive MIMO'')
under a pilot spoofing attack. Firstly, a simple method to
recognize the presence of this type of an active eavesdropping attack to a particular user
is shown. In order to deal with this attack, we consider the problem of  maximizing the achievable data rate of
the attacked user  or  its achievable secrecy rate.
The corresponding problems
of  minimizing the consumption power subject to security constraints are also considered in parallel.
Path-following algorithms are developed to solve the posed optimization problems under different
 power allocation to access points (APs). Under equip-power allocation to APs, these optimization problems admit
 closed-form solutions. Numerical results show their efficiencies.
	\end{abstract}

\keywords{Cell-free, channel estimation, pilot spoofing attack, active eavesdropping, inner convex approximation.}

\section{Introduction}\label{SEC: INTRO}
\subsection{Previous Works}

\subsubsection{Cell-free massive MIMO networks}
\emph{Cell-free massive MIMO} has been recently introduced in \cite{Ngo-2015-CellFree-SPAWC, Cell-free_AHien2017, AHien-2017-TGCN}. These papers showed that by proper implementation, cell-free massive MIMO can provide a uniformly good service to all users in the network and outperform small-cell massive MIMO in terms of throughput, and  handle the shadow fading correlation more efficiently. In a typical small-cell massive MIMO system, the channel from an access point (AP) to a user is a single scalar. In contrast, in a cell-free Massive MIMO system, all APs can liaise with each other via a central processing unit (CPU) to perform beamforming transmission tasks, and thus the effective channel (from an AP to a user) will take the form of an inner product between two vectors \cite{Cell-free_AHien2017}. That inner product can converge to its mean when the length of each vector (equivalently, the number of APs) is large enough. As a result, the effective channel also converges to a constant and there is no need to estimate downlink channels in the massive MIMO systems using cell-free architecture, while the small-cell counterpart may require both downlink and uplink training for channel estimation.

Inspired by \cite{Ngo-2015-CellFree-SPAWC, Cell-free_AHien2017, AHien-2017-TGCN},  cell-free massive MIMO  has been further studied in \cite{Buzzi-2017-CellFree-WCL, Liu-2017-CellFree-TSP, Long-2017-CellFree-COML,Toan-2017-CellFree-COML}.
Cell-free massive MIMO  was modified in \cite{Buzzi-2017-CellFree-WCL} to allow each AP serving only several users based on the strongest channels instead of serving all users.
The joint user association and interference/power control to mitigate the interference and cell-edge effect was considered
in \cite{Liu-2017-CellFree-TSP}.  The problem of designing zero-forcing precoders  to maximize the energy efficiency
 for cell-free massive MIMO networks was considered in \cite{Long-2017-CellFree-COML}. We are motivated to investigate the security aspect of cell-free massive MIMO as it was not considered in these papers.

\subsubsection{Pilot spoofing attack}
Recently, active eavesdropping has attracted the researchers' attention to physical layer security. It has been proved that active eavesdroppers are more dangerous than passive eavesdroppers because confidential information leaked to the active eavesdroppers is possibly higher \cite{Zhou-2012-Security-TWC}.
Active eavesdropping is an interesting topic which has been emerging in recent years. For instance,
active eavesdroppers are capable of jamming as well as eavesdropping \cite{ Abedi-2017-TWC, Li-2017-TIFS, Mukherjee-2013-Security-TSP} and/or they can send spoofing pilot sequences \cite{Amariucai-2012-TIFS, Zhou-2012-Security-TWC, Kapetanovic-2015-Security-Mag}. The latter scenario relates to the so-called \emph{pilot spoofing attacks} \cite{Amariucai-2012-TIFS,Zhou-2012-Security-TWC}. Eavesdropping attacks caused by an active eavesdropper is more harmful than passive ones. A feedback-based encoding scheme to  improve the  secrecy of transmission was
 proposed in  \cite{Amariucai-2012-TIFS}. On the contrary, from an eavesdropping point of view, \cite{Zhou-2012-Security-TWC} showed how an active eavesdropper achieves a satisfactory performance with the use of transmission energy.

Initialized  by \cite{Zhou-2012-Security-TWC},  pilot spoofing attacks in wireless security
have been actively studied  \cite{Kapetanovic-2015-Security-Mag, Wu-2016-Security-TIT, Im-2015-Security-TWC, Tugnait-2016-WCL, Xiong-2016-TIFS, Xiong-2015-TIFS}. By assuming that an eavesdropper can attack a wireless communication system during training phase to gain the amount of leaked information, the authors in \cite{Kapetanovic-2015-Security-Mag, Wu-2016-Security-TIT, Im-2015-Security-TWC, Tugnait-2016-WCL, Xiong-2016-TIFS, Xiong-2015-TIFS} have studied pilot contamination attacks in distinct scenarios. Their results reveal that active eavesdropping poses an actual threat to different types of wireless systems in general. More specifically, the authors in \cite{Kapetanovic-2015-Security-Mag} conducted a survey of detecting active attacks on massive MIMO systems. The authors in \cite{Wu-2016-Security-TIT} designed an artificial noise to cope with an active eavesdropper in a secure massive MIMO system. The use of artificial noise is not necessary in the present paper
 as our proposed optimization problems can also control beam steering towards intended destinations such that security constraints are met.
Meanwhile, the consideration of the authors in \cite{Im-2015-Security-TWC} is a secret key generation, which is beyond the scope of our paper. In \cite{Tugnait-2016-WCL} a method called \emph{minimum description length source enumeration} is employed to detect an active eavesdropping attack in a relaying network; however, the secure performance of the system (via metrics such as secrecy rate or secrecy outage probability) is not evaluated. Other detection techniques can be found in \cite{Xiong-2016-TIFS} and \cite{Xiong-2015-TIFS}. While \cite{Xiong-2016-TIFS} resort to the downlink phase to estimate channels and improve the system performance, we only use one training phase to detect a potential eavesdropper (which is presented in Appendix A). Our simple detection technique is similar to that in \cite{Xiong-2015-TIFS}, which also compares the asymmetry of received signal power levels to detect eavesdroppers. The differences between \cite{Xiong-2015-TIFS} and our paper lie in modelling (massive MIMO networks versus cell-free networks) and optimization formulations. Although the eavesdropping attack detection methods in \cite{Tugnait-2016-WCL, Xiong-2016-TIFS, Xiong-2015-TIFS} are really attractive, we will not delve into similar methods and not consider such a method as a major contribution.
Instead, we focus on solving optimization problems to provide specific solutions for cell-free systems in the case that a user is really suspected of being an eavesdropper.


\subsection{Contributions}



As discussed above, the introduction of a cell-free massive MIMO network can bring about a huge chance of improving throughput in comparison with small-cell networks. We thus study the security aspect of such a network and more importantly, this paper is the first work on the integration of security with the cell-free massive MIMO architecture. On the other hand, the analytical approach in this work is different from previous papers on security for massive MIMO. The major difference is that we do not use the law of large number to formulate approximate expressions for signal-to-noise (SNR) ratios. Instead, we consider lower- and upper- bounds for SNR expressions, thereby a lower-bound for secrecy rate is formulated and evaluated. This alternative approach, of course, holds true for general situations in which the number of nodes/antennas are not so many (and hence the term ``massive'' can be relatively understood and/or can be also removed).

In this paper, we examine a cell-free network in which an eavesdropper is actively involved in attacking the system during the training phase. We simply and shortly show that such an attack is dangerous but can be detected by a simple detection mechanism. Thereby, efforts to deal with active eavesdropping can be made and secure strategies can be prepared at APs during the next phase (i.e. the downlink phase). With these in mind and with the aim of keeping confidential information safe, we can realize beforehand which user is under attack and thus, we can propose optimization problems based on secrecy criteria to protect that user from being overhead. Our proposed optimization problems can be classified into 2 groups. For the first group, we design a matrix of power control coefficients
\begin{itemize}
	\item to maximize the achievable data rate of the user who is under attack (see \ref{subsec: P1})
	\item to maximize the achievable secrecy rate of user 1 (see \ref{subsec: Q1})
	\item to minimize the total power at all APs subject to the constraints on the data rate of each user, including all legitimate users and eavesdropper (see \ref{subsec: R1})
	\item to minimize the total power at all APs subject to the constraints on the achievable data rate and the data rates of other users (i.e. legitimate users not under attack) (see \ref{subsec: S1}).
\end{itemize}
For the second group, we design a common power control coefficient for all APs and consider 4 optimization problems (\ref{subsec: under P1}, \ref{subsec: under Q1}, \ref{subsec: under R1} and \ref{subsec: under S1}), which are similar and comparable to their counterparts in the first group.  While the common goal of all maximization programs is achievable secrecy rate, that of all minimization programs is power consumption at APs. Taking control of power at each AP, we find the most suitable solutions to the proposed optimization problems and compare them in secure performance as well as energy.

The rest of the paper is organized as follows. In Section II, the system model is presented. In Section III, we propose two maximization problems to maximize achievable secrecy rate subject to several quality-of-service constraints. In parallel, Section IV provides two minimization problems to minimize the power consumption such that security constraints are still guaranteed. In Section V, special cases of the proposed optimization problems are given for comparison purposes. Simulation results and conclusions are given in
Sections VI and VII, respectively.

\emph{Notation:} $[\cdot]^T$, $[\cdot]^*$, and $[\cdot]^{\dagger}$ denote the transpose operator, conjugate operator, and Hermitian operator, respectively. $[\cdot]^{-1}$ and $[\cdot]^+$ denote the inverse operator and pseudo-inverse operator, respectively. Vectors and matrices are represented with lowercase boldface and uppercase boldface, respectively. $\B{I}_n$ is the $n\times n$ identity matrix. $\|\cdot\|$ denotes the Euclidean norm. $\EX{\cdot}$ denotes expectation. $\B{z}\sim\CNn{n}{\B{\bar{z}},\B{\Sigma}}$ denotes a complex Gaussian vector $\B{z}\in\mathbb{C}^{n\times 1}$ with mean vector $\B{\bar{z}}$ and covariance matrix $\B{\Sigma}\in\mathbb{C}^{n\times n}$.


\begin{figure*}[t!]
   \centerline{\includegraphics[width=1\textwidth]{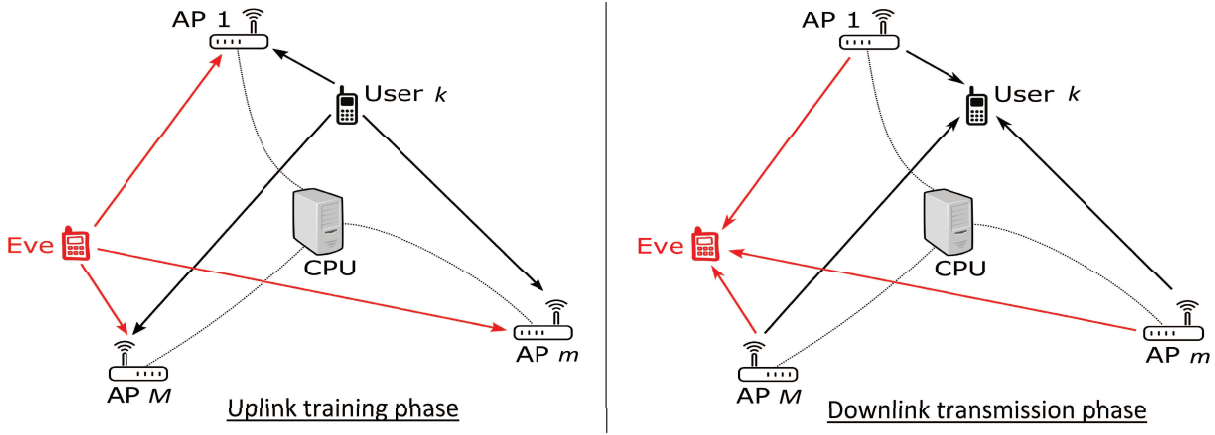}}
   \caption{
   	A system model consisting of $M$ APs, $K$ legal users and one active eavesdropper Eve. The arrows point to the direction from transmitters to receivers. All directions, connected to Eve, are in red. In uplink training phase, all users and Eve send the pilots to the APs in order to request for the messages, which privately intended for them. Connected together through a CPU, the APs exchange information, estimate channels and detect abnormality in pilot sequences. In downlink transmission phase, the APs transmit their designed signals to users and Eve.
   }
   \label{fig:number}
\end{figure*}

\section{Cell-Free System Model}\label{SEC: SYS}
We consider a system with $M$ APs and $K$ users in the presence of an active eavesdropper (Eve). Each node is equipped with a single antenna and all nodes are randomly positioned. Let $g_{mk} = \sqrt{\beta_{mk}} h_{mk} \sim\CN{0}{\beta_{mk}}$ be the downlink channel from the $m$th AP to the $k$th user.\footnote{In the formulation $g_{mk} = \sqrt{\beta_{mk}} h_{mk} $, the term $\beta_{mk}$ represents the large scale fading, while the term $h_{mk} \sim\CN{0}{1}$ implies the small scale fading. The value of $\beta_{mk}$ is constant and is based on a particular rule of power degradation. This rule will be presented in Section VI, given that the Hata-COST231 propagation prediction model is used (see \cite{Rappaport-Book-2002} and \cite{Chen-2006-VTC-HataModel}).} We assume channel reciprocity between uplink and downlink.
Similarly, let $g_{m\eav} \sim\CN{0}{\beta_{m\eav}}$ be the channel between the $m$th AP and Eve. Note that the desirable property of channel reciprocity requires the highly accurate calibration of hardware. In addition, the APs in cell-free massive MIMO systems are connected to a CPU via backhaul, thereby they can share information. We assume that the backhaul is perfect enough to consider error-free information only. Any limitation on capacity (caused by imperfect backhaul) will be left for future work.

The transmission includes 2 phases: Uplink training for channel estimation and downlink data transmission.

\subsection{Uplink training}
In this phase, the $k$th user sends a certain pilot vector $\B{p}_k\in\mathbb{C}^{T\times 1}$ to all APs where $T$ is an integer number. If $L_{int}$ denotes the coherence interval, then the first $T$ symbols are for pilot training and the $(L_{int} - T)$ remaining symbols are for data transmission. In low-mobility environment, the coherence interval can take on large numbers. It is shown that if the vehicle speed is $5.4$ km/h, the coherence interval $L_{int}$ can approach $15000$ symbols (see \cite[p.23] {Fundamentals_MassiveMIMO}). With such a large value of $L_{int}$, we can totally assign a sufficiently-large number to $T$ such that the inequality $T\geq K$ holds true. For example, $(T,K)=(150,100)$ is totally possible in practical situations (note that $T=150$ accounts for only $0.1\%$ of $L_{int}=15000$). In short, we can totally have $T\geq K$ and then design $K$ orthogonal pilot vectors such that $\B{p}_k^{\dagger} \B{p}_{k'}=0$ for $k\neq k'$ and $\|\B{p}_k\|^2=1$.
In general, $\B{p}_1, \ldots, \B{p}_K$ are known to Eve because the pilot sequences of a system are standardized and public. Taking advantage of this, Eve also sends its pilot sequence $\B{p}_{\eav}$ to all APs. If Eve wants to detect the signal destined for the $l$th user, $\B{p}_{\eav}$ will be designed to be the same as $\B{p}_l$ (see \cite{Zhou-2012-Security-TWC, Secure-Massive-PILOT_Kapetanovic2013, Secure-Massive-Pilot_Wu2015}).
Without the loss of generality, let us consider the situation in which Eve aims to overhear the confidential messages intended for the 1st user, i.e. $\B{p}_{\eav}=\B{p}_1$. At the $m$th AP, the received pilot vector is given by
\begin{align}\label{phase 1: signal 0}
  \B{y}_{p,m} = \sqrt{T\rho_u} \sum_{k=1}^{K} g_{mk} \B{p}_{k} + \sqrt{T\rho_{\eav}} g_{m\eav} \B{p}_{1} + \B{w}_m
\end{align}
where $\rho_u \triangleq P_u/N_0$ and $\rho_{\eav} \triangleq P_{\eav}/N_0$. Herein, $P_u$ and $P_{\eav}$ are the average transmit power of each user and that of Eve, respectively; while $N_0$ is the average noise power per a receive antenna. $\B{w}_m$ is an additive white Gaussian noise (AWGN) vector with $\B{w}_m\sim\CN{\B{0}}{\B{I}}$. Projecting $\B{y}_{p,m}$ onto $\B{p}_k^{\dagger}$, we can write the post-processing signal $y_{km}= \B{p}_k^{\dagger} \B{y}_{p,m}$ as\footnote{
If we assumed $T<K$ (i.e. $\B{p}_k^{\dagger} \B{p}_{k'}\neq0$ for $k\neq k'$), there would be the presence of the term $\sqrt{T\rho_u} \sum_{k'\neq k}^{K} g_{mk'} \B{p}_{k}^{\dagger} \B{p}_{k'}$ in \eqref{phase 1: signal 1}. Other changes could also be made and the framework of this paper could be re-applied.
}
\begin{align}\label{phase 1: signal 1}
  y_{km} =
    \left\{
      \begin{array}{ll}
        \sqrt{T\rho_u} g_{mk} + \B{p}_k^{\dagger} \B{w}_m, & k\neq 1 \\
        \sqrt{T\rho_u} g_{m1} + \sqrt{T\rho_{\eav}} g_{m\eav} + \B{p}_1^{\dagger}\B{w}_m, & k= 1
      \end{array}.
    \right.
\end{align}

It is of crucial importance that all APs are not aware of an eavesdropping attack until they have realized an abnormal sign from the sequence of signals $\{ y_{km} \}$ in \eqref{phase 1: signal 1}. Based on that abnormal sign, APs can identify the pilot which might be harmed. Therefore, it is necessary for APs to have a method to observe abnormality from $\{ y_{km} \}$. We describe such a method in Appendix \ref{Appendix: method}.

Besides, with the aim of estimating $g_{mk}$ and $g_{m\eav}$ from \eqref{phase 1: signal 1}, the MMSE method is adopted at the $m$th AP, i.e.
\begin{align}\label{phase 1: MMSE estimate 1}
\hat{g}_{mk} &=
\left\{
\begin{array}{ll}
\dfrac{\sqrt{T\rho_u}\beta_{mk}}{T\rho_u\beta_{mk} +1} y_{km}, & k\neq 1 \\
\dfrac{\sqrt{T\rho_u}\beta_{m1}}{T\rho_u\beta_{m1} + T\rho_{\eav}\beta_{m\eav} + 1} y_{1m}, & k= 1
\end{array}
\right.
\end{align}
and
\begin{align}\label{phase 1: MMSE estimate 2}
\hat{g}_{m\eav} &= \sqrt{\frac{ \rho_{\eav} }{\rho_u}} \frac{\beta_{m\eav}}{\beta_{m1}} \hat{g}_{m1}.
\end{align}
Let us denote
\begin{align}
\gamma_{mk}\triangleq \EX{|\hat{g}_{mk}|^2}
=\left\{
\begin{array}{ll}
\dfrac{T\rho_u \beta_{mk}^2}{T\rho_u \beta_{mk} + 1}, & k\neq 1 \\
\dfrac{T\rho_u \beta_{m1}^2}{T\rho_u \beta_{m1} + T\rho_{\eav} \beta_{m\eav} + 1}, & k= 1
\end{array}
\right.
\nonumber
\end{align}
and
$
\gamma_{m\eav}\triangleq \EX{|\hat{g}_{m\eav}|^2}.
$ Using \eqref{phase 1: MMSE estimate 2}, we can also rewrite
$$
\gamma_{m\eav} = \alpha_m \gamma_{m1}
$$ with
$
\alpha_m = \left( \rho_{\eav}\beta_{m\eav}^2 \right)/\left(\rho_u\beta_{m1}^2\right).
$
In association with the above, we state the following proposition for later use in the rest of paper.
\begin{proposition}
$\hat{g}_{mk}$ and $\hat{g}_{mk'}$ are uncorrelated for $\forall k'\neq k$. At the same time, $\hat{g}_{m\eav}$ and $\hat{g}_{mk'}$ are uncorrelated for $\forall k'\neq 1$. Furthermore, we have
\begin{align}\label{eq: pro 1}
    \EX{\left| \hat{g}_{mk} \hat{g}_{mk'}^* \right|^2}
        =
        \left\{
            \begin{array}{ll}
                \gamma_{mk} \gamma_{mk'}, & k'\neq k \\
                2 \gamma_{mk}^2, & k'=k
        \end{array},
        \right.
\end{align}
and
\begin{align}\label{eq: pro 2}
   \EX{\left| \hat{g}_{m\eav} \hat{g}_{mk'}^* \right|^2}
&= \left\{
\begin{array}{ll}
\alpha_m \gamma_{m1} \gamma_{mk'}, & k' \neq 1 \\
2 \alpha_m \gamma_{m1}^2, & k' = 1
\end{array}.
\right.
\end{align}
\end{proposition}
\begin{proof}
It is straightforward to prove the uncorrelated-ness by showing $\EX{\hat{g}_{mk} \hat{g}_{mk'}^*} = 0$ for $\forall k'\neq k$ and $\EX{\hat{g}_{m\eav} \hat{g}_{mk'}^*} =0$ for $\forall k'\neq 1$. Using these results, we can obtain \eqref{eq: pro 1} and \eqref{eq: pro 2} with the help of \eqref{phase 1: signal 1}--\eqref{phase 1: MMSE estimate 2} and the definitions of $\gamma_{mk}$ and $\gamma_{m\eav}$.
\end{proof}

Note that the eavesdropper's attack against the $1$st user during the training phase leads to the presence of $\rho_{\eav}$ in the denominator of $\hat{g}_{m1}$ (which is called a pilot spoofing attack).

\subsection{Downlink transmisson}
In this phase, the $m$th AP uses the estimate $\hat{g}_{mk}$ to perform beamforming technique. First, we denote $s_k$ be the signal intended for the $k$th user and $P_s$ be the average transmit power for a certain $s_k$. Then the signal transmitted by the $m$th AP can be designed (according to beamforming technique) as \cite{Cell-free_AHien2017}
\begin{align}\label{phase 2: beamforming 1}
  x_m &= \sqrt{P_s} \sum_{k=1}^{K} \sqrt{\eta_{mk}} \hat{g}_{mk}^* s_{k}
\end{align}
with $s_k$ being normalized such that $\EX{|s_k|^2}=1$. In \eqref{phase 2: beamforming 1}, $\eta_{mk}$ is the power control coefficient, which corresponds to the downlink channel from the $m$th AP to the $k$th user.

As such, the received signal at the $k$th user and Eve are, respectively, given by
\begin{align}
  z_k &= \sqrt{\rho_s}\sum_{m=1}^{M} g_{mk}
        \left(\sum_{k=1}^{K} \sqrt{\eta_{mk}} \hat{g}_{mk}^* s_{k}\right)
        + n_k ,
    \label{phase 2: signal 1}
    \\
  z_{\eav} &= \sqrt{\rho_s}\sum_{m=1}^{M} g_{m\eav}
        \left(\sum_{k=1}^{K} \sqrt{\eta_{mk}} \hat{g}_{mk}^* s_{k}\right)
        + n_{\eav}
    \label{Eve: signal 1}
\end{align}
where $\rho_s = P_s/N_0$, $n_k\sim\CN{0}{1}$, and $n_{\eav} \sim\CN{0}{1}$.

\subsubsection{The lower-bound for the mutual information between $s_k$ and $z_k$}
We rewrite \eqref{phase 2: signal 1} as
\begin{align}
  z_k &= \mathrm{DS}_k \times s_{k} +
    \underbrace{
    \mathrm{BU}_k \times s_{k} + \sum_{k'\neq k}^{K} \mathrm{UI}_{kk'} \times s_{k'} + n_k
    }_{\textrm{treated~as~aggregated~noise}} ,
  \label{phase 2: signal 2}
\end{align}
where
\begin{align}
	\mathrm{DS}_k &\triangleq \sqrt{\rho_s}\sum_{m=1}^{M} \EX{ \sqrt{\eta_{mk}} g_{mk} \hat{g}_{mk}^* } ,
	\nonumber \\
	\mathrm{BU}_k &\triangleq \sqrt{\rho_s} \sum_{m=1}^{M}
		\left(\sqrt{\eta_{mk}} g_{mk} \hat{g}_{mk}^* - \EX{ \sqrt{\eta_{mk}} g_{mk} \hat{g}_{mk}^* } \right) ,
	\nonumber \\
	\mathrm{UI}_{kk'} &\triangleq \sqrt{\rho_s} \sum_{m=1}^{M} \sqrt{\eta_{mk'}} g_{mk} \hat{g}_{mk'}^*
	\nonumber
\end{align}
represent the strength of the desired signal $s_k$, the beamforming gain uncertainty, and the interference caused by the $k'$th user (with $k'\neq k$), respectively. It is proved that the terms $\mathrm{DS}_k$, $\mathrm{BU}_k$, $\mathrm{UI}_{kk'}$ and $n_k$ in \eqref{phase 2: signal 2} are pair-wisely uncorrelated.

\begin{lemma}\label{lemma}
	Let $U$ and $V$ be complex-valued random variables with $U\sim\CN{0}{var\{U\}}$ and $\EX{|V|^2}=var\{V\}$. Given that $U$ and $V$ are uncorrelated, then the mutual information $I(U; U+V)$ between $U$ and $U + V$ is lower-bounded by $\log_2\left(1 + var\{U\} / var\{V\} \right)$. Consequently, the lower-bound SNR can be given by $var\{U\} / var\{V\} $.
\end{lemma}
\begin{proof}
	The reader is referred to \cite{Mutual_information-Shamai2002} and \cite{Mutual_Information-Goldsmith2006} for detailed proofs in terms of information theory.
\end{proof}

Let $I_k\left(s_k; z_k\right)$ denote the mutual information between $s_k$ and $z_k$. Considering the second, third, and fourth terms in \eqref{phase 2: signal 2} as noises, the lower-bound for $I_k\left(s_k; z_k \right)$ can be deduced from Lemma \ref{lemma} as follows:
\begin{align}\label{lower bound}
I_k\left(s_k; z_k \right)
\geq  \log_2 (1+\snr{k})
\end{align}
where
\begin{align}
\snr{k} &= \frac{ |\mathrm{DS}_k|^2 }{ \EX{|\mathrm{BU}_k|^2} + \sum_{k'\neq k}^{K} \EX{|\mathrm{UI}_{kk'}|^2} + 1 }
\nonumber
\\
&= \frac{ \rho_s \left(\sum_{m=1}^{M} \sqrt{\eta_{mk}} \gamma_{mk}\right)^2 }
{ \rho_s \sum_{k'=1}^{K} \sum_{m=1}^{M} \eta_{mk'} \gamma_{mk'} \beta_{mk} +1 }, ~k\in\mathcal{K}
\label{phase 2: SNR 2}
\end{align}
with $\mathcal{K}=\{1,2,\ldots,K\}$. The derivation of \eqref{phase 2: SNR 2} is available in \cite[Appendix A]{Cell-free_AHien2017}. The right hand side (RHS) of \eqref{lower bound} is the achievable data rate of user $k$.

\subsubsection{The upper-bound for the mutual information between $s_1$ and $z_{\eav}$}
We rewrite \eqref{Eve: signal 1} as
\begin{align}
  z_{\eav} &= \mathrm{BU}_{\eav, 1} \times s_1 +
    \underbrace{
     \sum_{k'\neq 1}^{K} \mathrm{UI}_{\eav, k'} \times s_{k'} + n_{\eav}
    }_{\textrm{treated~as~aggregated~noise}} .
  \label{Eve: signal 2}
\end{align}
where
\begin{align}
	\mathrm{BU}_{\eav, 1} &\triangleq \sqrt{\rho_s} \sum_{m=1}^{M}
		\sqrt{\eta_{m1}} g_{m\eav} \hat{g}_{m1}^*   ,
	\nonumber \\
	\mathrm{UI}_{\eav, k'} &\triangleq \sqrt{\rho_s} \sum_{m=1}^{M} \sqrt{\eta_{mk'}} g_{m\eav} \hat{g}_{mk'}^*
	\nonumber
\end{align}
respectively represent the strength of the desired signal $s_1$ (which Eve may want to overhear) and the interference caused by the remaining users (with $k'\neq k$). It is proved that the terms $\mathrm{BU}_{\eav, k}$, $\mathrm{UI}_{\eav, k k'}$ and $n_{\eav}$ in \eqref{Eve: signal 2} are pair-wisely uncorrelated. Thus, we can consider the second and third terms in \eqref{Eve: signal 2} as noises.

Let $I_{\eav} \left(s_1; z_{\eav} \right)$ denote the mutual information between $s_1$ and $z_{\eav}$. Then the upper-bound for $I_{\eav} \left(s_k; z_{\eav} \right)$ can be formulated as follows:
\begin{align}\label{upper bound}
I_{\eav} \left(s_1; z_{\eav} \right)
&\mathop\leq\limits^{(a)} I_{\eav} \left(s_1; z_{\eav} \left| \left\{g_{mk} \right\}_{m,k}, \left\{\hat{g}_{mk} \right\}_{m,k},  \left\{g_{m\eav}\right\}_m \right. \right)
\nonumber \\
&= \EXs{}{ \log_2
	\left( 1 +
	\frac{\left| \mathrm{BU}_{\eav, 1} \right|^2 }
	{  \sum_{k'\neq 1}^{K}  | \mathrm{UI}_{\eav, k'} |^2 + 1 }
	\right)
}
\nonumber \\
&\mathop\approx\limits^{(b)} \log_2 \left( 1 + \snrE{} \right)
\end{align}
where
\begin{align}
  \snrE{} &= \frac{ \EX{|\mathrm{BU}_{\eav,1}|^2} }{ \sum_{k'\neq 1}^{K} \EX{|\mathrm{UI}_{\eav, k'}|^2} + 1 }
  \label{phase 2: SNRe_1}
  \\
  &\mathop=\limits^{(c)}
\dfrac{
		\rho_s \sum_{m=1}^{M} \eta_{m1} \gamma_{m1}
		\left(
		\frac{ \rho_{\eav}\beta_{m\eav}^2 }{\rho_u\beta_{m1}^2}  \gamma_{m1}
		+ \beta_{m\eav}
		\right)
}{ \rho_s \sum_{k'\neq 1}^{K} \sum_{m=1}^{M} \eta_{mk'} \gamma_{mk'} \beta_{m\eav} + 1}.
\label{phase 2: SNRe_2}
\end{align}
The RHS of inequality $(a)$ means that Eve perfectly knows channel gains. It also implies the worst case in terms of security. Meanwhile, the approximation $(b)$ follows \cite[Lemma 1]{Zhang-2014-Rician-MassiveMIMO}. Finally, the derivation of $(c)$ is provided in Appendix B.

\subsubsection{Achievable secrecy rate}
From \eqref{lower bound} and \eqref{upper bound}, we can define the achievable secrecy rate of user $1$ as follows:
\begin{align}
\Delta &= I_1\left(s_1; z_1 \right)
- I_{\eav} \left(s_1; z_{\eav} \right)
\nonumber \\
&\geq  \log_2\left( ( 1+\snr{1} )/( 1+\snrE{} ) \right)
\triangleq R_{sec}
\end{align}
in which the explicit expressions for $\snr{1}$ and $\snr{\eav}$ are presented in \eqref{phase 2: SNR 2} and \eqref{phase 2: SNRe_2}, respectively.

%
%
%
%


In order to facilitate further analysis in the rest of paper, we denote $\B{\Psi}$ be the matrix in which the $(m,k)$th entry is $\B{\Psi}(m,k)=\sqrt{\eta_{mk}}$. The $k$th column vector of $\B{\Psi}$ is denoted as
\[
 \B{u}_k = \B{\Psi}(:,k)= \left[ \sqrt{\eta_{1k}}, \sqrt{\eta_{2k}},\ldots,\sqrt{\eta_{Mk}} \right]^T.
\]
Besides, we also define the following matrices and vectors
\begin{align}
	&\B{a}_k = \sqrt{\rho_s} \left[ \gamma_{1k}, \gamma_{2k}, \ldots, \gamma_{Mk} \right]^T,
	\nonumber \\
	&\B{A}_{kk'} =
	\sqrt{\rho_s} \diag{ \sqrt{\beta_{1k}\gamma_{1k'}}, \ldots,\sqrt{\beta_{Mk}\gamma_{Mk'}} } ,
	\nonumber \\
	&\B{B}_{\eav} = \sqrt{\rho_s}
	\diag{
		\sqrt{
			\gamma_{11}
			\left( \gamma_{1\eav} + \beta_{1\eav} \right)
		},
		\ldots,
		\sqrt{
			\gamma_{M1}
			\left( \gamma_{M\eav}	+ \beta_{M\eav} \right)
		}
	}
	\nonumber \\
	&\B{B}_{k'} =
	\sqrt{\rho_s} \diag{ \sqrt{\beta_{1\eav}\gamma_{1k'}}, \ldots,\sqrt{\beta_{M\eav}\gamma_{Mk'}} } \textrm{with~} k'\neq 1 .
	\nonumber
\end{align}
Finally, the SNRs in \eqref{phase 2: SNR 2} and \eqref{phase 2: SNRe_2} can be rewritten in a more elegant way as follows:
\begin{align}
	\snr{k} &=
	\left. \left( \B{a}_k^T \B{u}_k \right)^2 \right/
			 {\varphi_k(\B{\Psi}) },
	\\
	\snrE{}
	&=
	\left. \left\|\B{B}_{\eav}\B{u}_{1}\right\|^2 \right/
			{\varphi_{\eav}(\B{\Psi}) }	
\end{align}
where
\begin{align}
\varphi_k(\B{\Psi})
	&= \sum_{k'=1}^{K} \left\| \B{A}_{kk'} \B{u}_{k'} \right\|^2 + 1, ~k\in\mathcal{K},
\\
\varphi_{\eav}(\B{\Psi})
	&= \sum_{k'\neq 1}^{K} \left\| \B{B}_{k'} \B{u}_{k'} \right\|^2 + 1.
\end{align}
All SNR-related expressions are now presented as functions of $\Psi$ instead of $\{\eta_{mk}\}_{m,k}$. Given that $\eta_{mk}$ decides the amount of the $m$th AP's power  destined for the $k$ user, the $(m,k)$th entry of $\Psi$ is also referred to as the factor deciding how much transmit power used by the $m$th AP and destined for the $k$ user.

\section{Secrecy Rate Maximization}\label{SEC: PRO 1}
In this section, we aim to design the matrix $\B{\Psi}$ to maximize either the achievable data rate of user 1 (in nats/s/Hz), i.e. $\ln\left(1+\snr{1}\right)$, or its achievable
secrecy rate $\ln\left(1 + \snr{1}\right)-\ln\left( 1 + \snrE{1} \right)$ in improving the secure performance of our system. Prior to performing these tasks, however, we need to impose a critical condition on the power at each AP. The power constraint is described as follows:
\begin{itemize}
\item
Let $P_{max}$ be the maximum transmit power of each AP, i.e. $P_{max} \geq \EX{|x_m|^2}$. From \eqref{phase 2: beamforming 1}, the average transmit power for the $m$th AP can be given by
\begin{align}\label{Average Power}
\EX{|x_m|^2} = P_s \sum_{k=1}^{K} \eta_{mk} \gamma_{mk}.
\end{align}
With the power constraint on every AP, we have
\begin{align}
 \sum_{k=1}^{K} \B{\Psi}^2(m,k) \gamma_{mk}  &\leq \frac{ \rho_{max} }{ \rho_s } , ~ m\in\mathcal{M}
 \label{Power Constraint}
\end{align}
with $\mathcal{M}=\{1,\ldots,M\}$. Note that $\rho_{max} = P_{max}/N_0$ is viewed as the maximum possible ratio of the $m$th AP's average transmit power to the average noise power.
\end{itemize}

Now we begin with optimizing $\B{\Psi}$ to maximize the achievable data rate of the $1$st user (who is under attack), i.e.
\begin{subequations}\label{Problem P1}
\begin{eqnarray}
(\textbf{P1})
~~&\ds\max_{ \B{\Psi} }
~~& \ln \left( 1 + \left. \left( \B{a}_1^T \B{u}_1 \right)^2 \right/ { \varphi_1(\B{\Psi}) } \right)
\label{Problem P1: C1} \\
~~& \text{s.t.}
~~& \eqref{Power Constraint},
\label{P1a}\\
~~&& \frac{
 \left\|\B{B}_{\eav}\B{u}_{1}\right\|^2
}{\varphi_{\eav}(\B{\Psi})} \leq \theta_{\eav},
\label{P1b} \\
~~&& \ds\frac{ \left( \B{a}_k^T \B{u}_k \right)^2 }{ \varphi_k(\B{\Psi}) } \geq \theta_k,  ~k\in\mathcal{K}\backslash\{1\} .
\label{P1c}
\end{eqnarray}
\end{subequations}
Herein, optimizing $\B{\Psi}$ is equivalent to finding the optimal value of every power control coefficient $\eta_{mk}$ (because of the relation $\B{\Psi}(m,k) = \sqrt{\eta_{mk}}$).

The constraint \eqref{Power Constraint} is to control the transmit power at each AP as previously described. The constraint \eqref{P1b} requires that the \emph{greatest} amount of information Eve can captures will not exceed some predetermined threshold, i.e. $\ln\left(1+\snrE{}\right) \leq \ln(1+\theta_{\eav})$. Finally, the constraint \eqref{P1c} guarantees that the achievable data rate of user $~k\in\mathcal{K}\backslash\{1\}$ is equal to or greater than some target threshold, i.e. $\ln\left(1+\snr{k}\right) \geq \ln(1+\theta_k)$.

Similarly, we will optimize every $\eta_{mk}$ (through optimizing the coefficient matrix $\B{\Psi}$) to maximize the achievable secrecy rate of user 1, i.e.
\begin{subequations}\label{Q1}
\begin{eqnarray}
(\textbf{Q1})
&\ds\max_{ \B{\Psi} }
&
\ln \left(
\frac{
1 +  \left. \left( \B{a}_1^T \B{u}_1 \right)^2 \right/ { \varphi_1(\B{\Psi}) }
}{
1 + \left. \left\|\B{B}_{\eav}\B{u}_{1}\right\|^2 \right/
			{\varphi_{\eav}(\B{\Psi}) }
}
\right)
\label{Q1a}\\
&\text{s.t.}	
& \eqref{Power Constraint}, \eqref{P1c}.
\end{eqnarray}
\end{subequations}

It should be noted that both problems {\bf (P1)} and {\bf (Q1)} has been considered in \cite{NTDP17a} and
\cite{NTDP17b} in the context of conventional MIMO systems, information and energy transfer. Inspired by these two works, we also use path-following algorithms to solve non-convex optimization problems. As can be seen in the subsections below, each of the proposed path-following algorithms invokes only one simple convex quadratic program at each iteration and thus, at least a locally optimal solution can be found out.

\subsection{Solving  problem $(\textbf{P1})$}\label{subsec: P1}
We can see that the constraint \eqref{Power Constraint} is obviously convex, while (\ref{P1c}) is
the following second-order cone (SOC) constraint and thus convex:
\begin{equation}\label{C4}
\frac{1}{\sqrt{ \theta_k }} \B{a}_k^T \B{u}_k\geq  \sqrt{\varphi_k(\B{\Psi})},
	~k\in\mathcal{K}\backslash\{1\}.
\end{equation}
Besides, we observe that the objective function of (\textbf{P1}) can be replaced with $\left. \left( \B{a}_1^T \B{u}_1 \right)^2 \right/ { \varphi_1(\B{\Psi}) }$.
Let $\B{\Psi}^{(\kappa)}$ be a feasible point for (\textbf{P1}) found from the $(\kappa-1)$th iteration. By using the inequality
\begin{eqnarray}
\frac{x^2}{y}\geq 2\frac{\bar{x}}{\bar{y}}x-\frac{\bar{x}^2}{\bar{y}^2}y\quad
\forall\ x>0, y>0, \bar{x}>0, \bar{y}>0\label{ine1}
\end{eqnarray}
we obtain
\begin{equation}\label{ine3}
\frac{ \left( \B{a}_1^T \B{u}_1 \right)^2 }{ \varphi_1(\B{\Psi}) }
\geq f_1^{(\kappa)}(\B{\Psi})
\triangleq a^{(\kappa)}\B{a}_1^T \B{u}_1-b^{(\kappa)}\varphi_1(\B{\Psi})
\end{equation}
with
\begin{equation}\label{ine4}
a^{(\kappa)}=2\frac{ \left( \B{a}_1^T \B{u}_1^{(\kappa)} \right)^2 }{ \varphi_1(\B{\Psi}^{(\kappa)}) }
,~
b^{(\kappa)}=(a^{(\kappa)}/2)^2.
\end{equation}
As such, maximizing $\left. \left( \B{a}_1^T \B{u}_1 \right)^2 \right/ { \varphi_1(\B{\Psi}) }$ is now equivalent to maximizing $f_1^{(\kappa)}(\B{\Psi})$.
Finally, considering the function $\varphi_{\eav} (\B{\Psi})$ in (\ref{P1b}), we find that it is convex quadratic and thus, the non-convex constraint (\ref{P1b}) is innerly approximated by the convex quadratic constraint\footnote
{The right hand side of \eqref{snrE} is the first-order Taylor approximation of $\varphi_{\eav}(\B{\Psi})$ near $\B{\Psi}^{(\kappa)}$. With $\varphi_{\eav}(\B{\Psi})$ being convex, we have $\varphi^{(\kappa)}_{\eav}(\B{\Psi}) \leq \varphi_{\eav}(\B{\Psi})$.
}
\begin{equation}
	\Bigl. \left\|\B{B}_{\eav}\B{u}_{1}\right\|^2 \Bigr/ \theta_{\eav}
    \leq  \varphi_{\eav}^{(\kappa)}(\B{\Psi})
    \label{snrE}
\end{equation}
for
\begin{equation}\label{snrEa}
\varphi_{\eav}^{(\kappa)}(\B{\Psi})\triangleq \sum_{k\neq 1}^{K} \left[ { \B{u}_k^{(\kappa)} }^T \B{B}_k^2 \left( 2\B{u}_k - \B{u}_k^{(\kappa)} \right) \right] + 1.
\end{equation}
Having the approximations (\ref{ine3}) and (\ref{snrE}), at $\kappa$-th iteration we solve the
following convex optimization to generate a feasible point $\B{\Psi}^{(\kappa+1)}$:
\begin{equation}\label{Problem P4}
\max_{ \B{\Psi}}\ \ f_1^{(\kappa)}(\B{\Psi})\quad
\text{s.t.}\quad  \eqref{Power Constraint}, (\ref{C4}), (\ref{snrE}).
\end{equation}
The problem \eqref{Problem P4} involves $MK$ scalar real variables (because $\B{\Psi}$ has $MK$ entries) and $\epsilon = M+K$ quadratic constraints. According to \cite{NTDP17b}, the per-iteration cost to solve \eqref{Problem P4} is $\mathcal{O}\left( (MK)^2 \epsilon^{2.5} + \epsilon^{3.5} \right)$.

To find a feasible point for (\textbf{P1}) to initialize the above procedure, we address the problem
\begin{equation}\label{ini1}
	\min_{\B{\Psi}}\
	\Bigl.
	\left\|\B{B}_{\eav}\B{u}_{1}\right\|^2
	\Bigr/ \theta_{\eav}-\varphi_{\eav}(\B{\Psi})\quad\mbox{s.t.}\quad
 \eqref{Power Constraint}, (\ref{C4}).
\end{equation}
Initialized by any feasible point $\B{\Psi}^{(0)}$ for convex constraints  \eqref{Power Constraint} and
(\ref{C4}), we iterate the following optimization problem
\begin{eqnarray}\label{ini2}
	\min_{\B{\Psi}}\
	\Bigl. \left\|\B{B}_{\eav}\B{u}_{1}\right\|^2
	\Bigr/ \theta_{\eav}-\varphi_{\eav}^{(\kappa)}(\B{\Psi})
\quad\mbox{s.t.}\quad \eqref{Power Constraint}, (\ref{C4}),
\end{eqnarray}
till
\begin{equation}\label{ini3}
	\Bigl.
	\left\|\B{B}_{\eav}\B{u}_{1}^{(\kappa)}\right\|^2
	\Bigr/ \theta_{\eav} - \varphi_{\eav}\left(\B{\Psi}^{(\kappa)}\right) \leq 0,
\end{equation}
so $\B{\Psi}^{(\kappa)}$ is feasible for (\textbf{P1}).
To sum up, we provide the  following algorithm:

\begin{algorithm}[h!]
	\caption{Path-following algorithm for solving $(\textbf{P1})$}
	\begin{algorithmic}[1]
		\STATE \textbf{Initialization}: Set $\kappa = 0$ with a feasible point
		$\B{\Psi}^{(0)}$ for $(\textbf{P1})$.
		\REPEAT
		\STATE Solve \eqref{Problem P4} to obtain the optimal solution $\B{\Psi}^{(\kappa+1)}$.
		\STATE Reset $\kappa := \kappa + 1$.
		\UNTIL Converge.
		\RETURN $\B{\Psi}^{(\kappa)}$ as the desired result.
	\end{algorithmic}
\end{algorithm}

\subsection{Solving  problem $(\textbf{Q1})$}\label{subsec: Q1}
By using the inequality \cite{H.Tuy-Book-2016-Opt}
\begin{align}
\ln \left(1+\frac{x^2}{y}\right)
&\geq
\ln\left(1+\frac{\bar{x}^2}{\bar{y}}\right)
+
\frac{ \frac{\bar{x}^2 }{ \bar{y} } }
	{ 1 + \frac{ \bar{x}^2 }{ \bar{y} } }
\left( 2 - \frac{\bar{x}}{2x-\bar{x}}-\frac{y}{\bar{y}} \right)
\nonumber \\
&\textrm{for~} \forall\ x>0, \bar{x}>0, y>0, \bar{y}>0, 2x>\bar{x}
\label{ine6}
\end{align}
we obtain
\begin{align}
&\ln\left( 1 +  \frac{ \left( \B{a}_1^T \B{u}_1 \right)^2 }{ \varphi_1(\B{\Psi}) } \right)
\nonumber\\ &\geq
a^{(\kappa)}+b^{(\kappa)}\left(2-\ds\frac{\varphi_1(\B{\Psi})}{\varphi_1(\B{\Psi}^{(\kappa)})}\right.
\left.-\ds\frac{( \B{a}_1^T \B{u}_1^{(\kappa)})^2}{2\B{a}_1^T \B{u}_1^{(\kappa)}\B{a}_1^T \B{u}_1
-(\B{a}_1^T \B{u}_1^{(\kappa)})^2}\right)
\nonumber\\ &\triangleq
f^{(\kappa)}(\B{\Psi})
\label{ine8}
\end{align}
over the trust region
\begin{equation}\label{ine9}
2\B{a}_1^T \B{u}_1^{(\kappa)}\B{a}_1^T \B{u}_1
-(\B{a}_1^T \B{u}_1^{(\kappa)})^2>0
\end{equation}
for
\begin{align}
a^{(\kappa)} &= \ln \left(1+t^{(\kappa)}\right) ,
\nonumber \\
b^{(\kappa)} &= \Bigl. t^{(\kappa)} \Bigr/ \left(1+t^{(\kappa)}\right),
\nonumber \\
t^{(\kappa)} &= \Bigl. \left( \B{a}_1^T \B{u}_1^{(\kappa)} \right)^2 \Bigr/ \varphi_1\left(\B{\Psi}^{(\kappa)}\right).
\nonumber
\end{align}
In addition, by respectively using the inequality \cite{H.Tuy-Book-2016-Opt}
\begin{equation}\label{ine7}
	\ln(1+x)\leq \ln(1+\bar{x})-\frac{\bar{x}}{1+\bar{x}}+\frac{x}{\bar{x}+1}, ~\forall\ x>0, \bar{x}>0
\end{equation}
and the fact that $\varphi^{(\kappa)}_{\eav}(\B{\Psi}) \leq \varphi_{\eav}(\B{\Psi})$ (please see Footnote 2), we obtain
\begin{align}
\ln\left( 1 +\dfrac{ \left\|\B{B}_{\eav}\B{u}_{1}\right\|^2 }
			{\varphi_{\eav}(\B{\Psi}) }\right)
&\mathop\leq\limits^{} c^{(\kappa)}
+ d^{(\kappa)}\ds\frac{
	\left\|\B{B}_{\eav}\B{u}_{1}\right\|^2
}
{\varphi_{\eav}(\B{\Psi})}
\nonumber \\
&\leq
c^{(\kappa)}+d^{(\kappa)}\ds\frac{
	\left\|\B{B}_{\eav}\B{u}_{1}\right\|^2
}
{\varphi_{\eav}^{(\kappa)}(\B{\Psi})}
\triangleq
g^{(\kappa)}(\B{\Psi})
\end{align}
over the trust region
\begin{equation}\label{ine10}
\varphi_{\eav}^{(\kappa)}(\B{\Psi})>0
\end{equation}
for
\begin{align}
c^{(\kappa)} &=
	\ln(1+t^{(\kappa)}_{\eav})- \Bigl. t^{(\kappa)}_{\eav} \Bigr/ \left(1+t^{(\kappa)}_{\eav}\right),
\nonumber \\
d^{(\kappa)} &=
	1 \Bigl. \Bigr/\left( 1+t^{(\kappa)}_{\eav} \right),
\nonumber \\
t^{(\kappa)}_{\eav} &=
	\Bigl.
	\left\|\B{B}_{\eav}\B{u}_{1}^{(\kappa)}\right\|^2
	\Bigr/ \varphi_{\eav}\left(\B{\Psi}^{(\kappa)}\right) .
\nonumber
\end{align}

Initialized by a feasible point $\B{\Psi}^{(0)}$ for the convex constraints \eqref{Power Constraint} and (\ref{C4}),
at $\kappa$-th iteration for $\kappa=0, 1,\dots,$ we solve the following convex optimization problem to generate the
next feasible point $\B{\Psi}^{(\kappa+1)}$:
\begin{subequations}\label{ine11}
\begin{eqnarray}
& \ds\max_{\B{\Psi}} & f^{(\kappa)}(\B{\Psi})-g^{(\kappa)}(\B{\Psi})
\\
& \mbox{s.t.} &
\eqref{Power Constraint}, (\ref{C4}), (\ref{ine9}), (\ref{ine10}).
\end{eqnarray}
\end{subequations}
With $MK$ scalar real variables, $2$ linear constraints and $(\epsilon - 1)$ quadratic constraints, the per-iteration cost to solve \eqref{ine11} is $\mathcal{O}\left( (MK)^2 (\epsilon - 1)^{2.5} + (\epsilon - 1)^{3.5} \right)$.

As such, the problem $(\textbf{Q1})$ can be solved by using the following algorithm:

\begin{algorithm}[h!]
	\caption{Path-following algorithm for solving $(\textbf{Q1})$}
	\begin{algorithmic}[1]
		\STATE \textbf{Initialization}: Set $\kappa = 0$ with a feasible point
		$\B{\Psi}^{(0)}$ for $(\textbf{Q1})$.
		\REPEAT
		\STATE Solve \eqref{ine11} to obtain the optimal solution $\B{\Psi}^{(\kappa+1)}$.
		\STATE Reset $\kappa := \kappa + 1$.
		\UNTIL Converge.
		\RETURN $\B{\Psi}^{(\kappa)}$ as the desired result.
	\end{algorithmic}
\end{algorithm}

\section{Power Minimization}\label{SEC: PRO 3}
In this section, we aim to design the matrix $\B{\Psi}$ to minimize the total average transmit power of all APs subject to security constraints as well as other SNR-based constraints:
\begin{subequations}\label{Problem R1}
\begin{eqnarray}
	(\textbf{R1}) &\ds \min_{\B{\Psi}}  & \sum_{m=1}^M\sum_{k=1}^{K} \B{\Psi}^2(m,k) \gamma_{mk}
	\\
	& \mbox{s.t.}
	& \eqref{Power Constraint}, (\ref{P1b}),
	\\
	&& \ds\frac{ \left( \B{a}_k^T \B{u}_k \right)^2 }{ \varphi_k(\B{\Psi}) } \geq \theta_k,  ~k\in\mathcal{K}
	\label{R1c}
\end{eqnarray}
\end{subequations}
and
\begin{subequations}\label{Problem S1}
\begin{eqnarray}
	(\textbf{S1}) &\ds \min_{\B{\Psi}}  & \sum_{m=1}^M\sum_{k=1}^{K} \B{\Psi}^2(m,k) \gamma_{mk}
	\\
	& \mbox{s.t.}
	& \eqref{Power Constraint}, \eqref{C4},
	\label{S1a}
	 \\
	&&\ln \left(
	\frac{
		1 +  \left. \left( \B{a}_1^T \B{u}_1 \right)^2 \right/ { \varphi_1(\B{\Psi}) }
	}{
		1 + \left. \left\|\B{B}_{\eav}\B{u}_{1}\right\|^2 \right/
		{\varphi_{\eav}(\B{\Psi}) }
	}
	\right)
	\geq r_{\phi}.
	\label{S1b}	
\end{eqnarray}
\end{subequations}
	Again, $\B{\Psi}(m,k)$ is the $(m,k)$th entry of the matrix $\B{\Psi}$. Due to the relation $\B{\Psi}(m,k) = \sqrt{\eta_{mk}}$, finding $\B{\Psi}$ is equivalent to finding every power control coefficient $\eta_{mk}$ ($m\in\mathcal{M}$ and $k\in\mathcal{K}$).

In addition, the objective function is the total power \emph{radiated} by the antennas of APs. The power consumed by other components (such as the backhaul and the CPU) is beyond the scope of this paper.

Note that \eqref{R1c} is not exactly the same as \eqref{C4} because \eqref{R1c} contains one more constraint, i.e. $\snr{1} \geq \theta_1$. Meanwhile, $r_{\phi}$ in the program (\textbf{S1}) is the given threshold which a designer may want to obtain. In general, we will have different results (which of course leads to different secure performances) when using (\textbf{R1}) and (\textbf{S1}). However, the obtained results can also be the same when using these programs, depending on the given values of $\theta_1$, $\theta_{\eav}$ and $r_{\phi}$.

\subsection{Solving  problem $(\textbf{R1})$}\label{subsec: R1}
At $\kappa$-th iteration, we solve the following convex optimization problem to generalize the next iterative
feasible point $\B{\Psi}^{(\kappa+1)}$
\begin{subequations}\label{R1.a}
\begin{eqnarray}
&\ds\min_{\B{\Psi}}  &\sum_{m=1}^M\sum_{k=1}^{K} \B{\Psi}^2(m,k) \gamma_{mk} \\
& \mbox{s.t.} &
\eqref{Power Constraint}, (\ref{C4}), (\ref{snrE}).
\end{eqnarray}
\end{subequations}
Similar to \eqref{Problem P4}, the computational complexity of solving \eqref{R1.a} is also $\mathcal{O}\left( (MK)^2 \epsilon^{2.5} + \epsilon^{3.5} \right)$.

Note that a feasible point $\B{\Psi}^{(0)}$ for (\textbf{R1}) can be found in the same way as $(\textbf{P1})$. Furthermore, the algorithm for solving (\textbf{R1}) is presented below.

\begin{algorithm}[h!]
	\caption{Path-following algorithm for solving $(\textbf{R1})$}
	\begin{algorithmic}[1]
		\STATE \textbf{Initialization}: Set $\kappa = 0$ with a feasible point
		$\B{\Psi}^{(0)}$ for $(\textbf{R1})$.
		\REPEAT
		\STATE Solve \eqref{R1.a} to obtain the optimal solution $\B{\Psi}^{(\kappa+1)}$.
		\STATE Reset $\kappa := \kappa + 1$.
		\UNTIL Converge.
		\RETURN $\B{\Psi}^{(\kappa)}$ as the desired result.
	\end{algorithmic}
\end{algorithm}

\subsection{Solving  problem $(\textbf{S1})$}\label{subsec: S1}
At $\kappa$-th iteration, we solve the following convex optimization problem to generalize the next iterative
feasible point $\B{\Psi}^{(\kappa+1)}$:
\begin{subequations}\label{S1.a}
\begin{eqnarray}
&\ds\min_{\B{\Psi}}  &\sum_{m=1}^M\sum_{k=1}^{K} \B{\Psi}^2(m,k) \gamma_{mk}
\\
&\mbox{s.t.}
&\eqref{Power Constraint}, (\ref{C4}),
	 \label{S1.aa}\\
&&f^{(\kappa)}(\B{\Psi})-g^{(\kappa)}(\B{\Psi})\geq r_{\phi} .
\label{S1.ab}
\end{eqnarray}
\end{subequations}
Similar to \eqref{Problem P4} and \eqref{R1.a}, the computational complexity of solving \eqref{S1.a} is also $\mathcal{O}\left( (MK)^2 \epsilon^{2.5} + \epsilon^{3.5} \right)$.

Note that a feasible point $\B{\Psi}^{(0)}$ for (\textbf{S1}) can be found like that for $(\textbf{Q1})$. Finally, we provide the detailed algorithm for solving (\textbf{S1}) as follows:

\begin{algorithm}[h!]
	\caption{Path-following algorithm for solving $(\textbf{S1})$}
	\begin{algorithmic}[1]
		\STATE \textbf{Initialization}: Set $\kappa = 0$ with a feasible point
		$\B{\Psi}^{(0)}$ for $(\textbf{S1})$.
		\REPEAT
		\STATE Solve \eqref{S1.a} to obtain the optimal solution $\B{\Psi}^{(\kappa+1)}$.
		\STATE Reset $\kappa := \kappa + 1$.
		\UNTIL Converge.
		\RETURN $\B{\Psi}^{(\kappa)}$ as the desired result.
	\end{algorithmic}
\end{algorithm}

\section{Optimization under Equal Power Allocation at Access Points}\label{SEC: SPECIAL CASESZ}
In this section, we reconsider the proposed optimization problems with $\eta_{mk}$ being equal to $\eta$ (for all $m$ and $k$) for comparison purposes.

Plugging $\eta_{mk} = \eta$ into \eqref{phase 2: SNR 2}--\eqref{phase 2: SNRe_2}, we obtain the special expressions for $\snr{k}$ and $\snr{\eav}$ as follows:
\begin{align}
	\left. \snr{k} \right|_{\eta_{mk} = \eta}
	&=  \eta \omega_k / \left( \eta \breve{\omega}_k +1 \right) ,
	\\
	\left. \snrE{} \right|_{\eta_{mk} = \eta}
	&=
	\dfrac{
			\eta \varpi
	}{ \eta \breve{\varpi} + 1}
\end{align}
where
\begin{align}
\omega_k &= \rho_s \left(\sum_{m=1}^{M} \gamma_{mk}\right)^2 ,
\nonumber \\
\breve{\omega}_k &= \rho_s \sum_{k'=1}^{K} \sum_{m=1}^{M} \gamma_{mk'} \beta_{mk} ,
\nonumber \\
\varpi &= \rho_s \sum_{m=1}^{M}  \gamma_{m1}
\left(
\frac{ \rho_{\eav}\beta_{m\eav}^2 }{\rho_u\beta_{m1}^2}  \gamma_{m1}
+ \beta_{m\eav}
\right) ,
\nonumber \\
\breve{\varpi} &= \rho_s \sum_{k'\neq 1}^{K} \sum_{m=1}^{M} \gamma_{mk'} \beta_{m\eav}.
\nonumber
\end{align}
Then,  problems (\textbf{P1}) and (\textbf{Q1}) reduce to
\begin{subequations}\label{Problem P1: special}
	\begin{eqnarray}
	(\underline{\textrm{P1}})
	&\ds\max_{ \eta }
	& \eta \omega_1 / \left( \eta \breve{\omega}_1 +1 \right)
	\\
	& \text{s.t.}
	& \eta \leq \frac{ \rho_{max}/\rho_s }{  \sum_{k=1}^K \gamma_{mk} } , ~m\in\mathcal{M}
	\label{special P1: b} \\
	&& \eta \left(\varpi  - \theta_{\eav} \breve{\varpi} \right) \leq \theta_{\eav} ,
	\label{special P1: c} \\
	&& \eta \left(\omega_k - \theta_k \breve{\omega}_k \right) \geq \theta_k,  ~k\in\mathcal{K}\backslash\{1\}
	\label{special P1: d}
	\end{eqnarray}
\end{subequations}
and
\begin{subequations}\label{Problem Q1: special}
	\begin{eqnarray}
	(\underline{\textrm{Q1}})
	~&\ds\max_{ \eta }
	~& \left( 1 + \frac{ \eta \omega_1 }{  \eta \breve{\omega}_1 +1 } \right)
	\left/
	 \left( 1 + \frac{ \eta \varpi }{ \eta \breve{\varpi} + 1 } \right)
	\right.
	\\
	~& \text{s.t.}
	~& \eqref{special P1: b}, \eqref{special P1: d}.
	\label{special Q1: b}
	\end{eqnarray}
\end{subequations}
Similarly,  problems (\textbf{R1}) and (\textbf{S1}) reduce to
\begin{subequations}\label{Problem R1: special}
	\begin{eqnarray}
	(\underline{\textrm{R1}}) &\ds \min_{\eta}  & \eta
	\\
	& \mbox{s.t.}
	& \eqref{special P1: b}, \eqref{special P1: c},
	\\
	&& \frac{ \eta \omega_k }{ \left( \eta \breve{\omega}_k +1 \right) } \geq \theta_k, ~k\in\mathcal{K}
	\label{special R1: c}
	\end{eqnarray}
\end{subequations}
and
\begin{subequations}\label{Problem S1: special}
	\begin{eqnarray}
	(\underline{\textrm{S1}}) &\ds \min_{\eta}  & \eta
	\\
	& \mbox{s.t.}
	& \eqref{special P1: b}, \eqref{special P1: d},
	\\
	&& \frac{  1 + \eta \omega_1 / \left( \eta \breve{\omega}_1 +1 \right)
	}
	{ 1 + \eta \varpi / \left( \eta \breve{\varpi} + 1 \right)
	}
	\geq \phi.
	\label{special S1: c}
	\end{eqnarray}
\end{subequations}


\subsection{Closed-form solutions to (\underline{\textrm{P1}})}\label{subsec: under P1}
The objective function of (\underline{\textrm{P1}}) increases in $\eta$. Hence, maximizing that objective function is equivalent to maximizing $\eta$. In other words, we will solve the following problem

\begin{subequations}
	\begin{eqnarray}
	(\underline{\textrm{P1}})~
	&\ds\max_{ \eta }
	& \eta
	\\
	& \text{s.t.}
	& \eqref{special P1: b}, \eqref{special P1: c}, \eqref{special P1: d}.
	\end{eqnarray}
\end{subequations}
In order for \eqref{special P1: d} to be meaningful, we need the condition
\begin{align}\label{theta_k: necessary cond. 1}
\left(\omega_k - \theta_k \breve{\omega}_k \right) > 0
\Leftrightarrow
\theta_k < \omega_k / \breve{\omega}_k
\end{align}
with $k\in\mathcal{K}\backslash\{1\}$. If $\theta_k$ satisfies the above condition, we can infer from both \eqref{special P1: b} and \eqref{special P1: d} the following:
\begin{align}\label{theta_k: necessary cond. 2}
\underbrace{
	\max_{ k\in\mathcal{K}\backslash\{1\} } \left\{ \frac{\theta_k}{ \omega_k - \theta_k \breve{\omega}_k } \right\}
}_{\geq 0}
\leq \eta \leq
\underbrace{
\min_{m\in\mathcal{M}} \left\{ \frac{ \rho_{max}/\rho_s }{  \sum_{k=1}^K \gamma_{mk} } \right\}
}_{> 0} .
\nonumber
\end{align}
This also implies another necessary condition as follows:
\begin{align}
\theta_k <
\frac{
	\omega_k \min_{m\in\mathcal{M}} \left\{ \frac{ \rho_{max}/\rho_s }{  \sum_{k=1}^K \gamma_{mk} } \right\}
}{
	1 + \breve{\omega}_k
	\min_{m\in\mathcal{M}} \left\{ \frac{ \rho_{max}/\rho_s }{  \sum_{k=1}^K \gamma_{mk} } \right\}
}
\end{align}
for each $k\in\mathcal{K}\backslash\{1\}$. The two conditions \eqref{theta_k: necessary cond. 1} and \eqref{theta_k: necessary cond. 2} are now rewritten in the following form:
\begin{align}\label{nec. cond.}
\theta_k &<
\min \left\{
\frac{\omega_k }{ \breve{\omega}_k } ,
\frac{
	\omega_k \min_{m\in\mathcal{M}} \left\{ \frac{ \rho_{max}/\rho_s }{  \sum_{k=1}^K \gamma_{mk} } \right\}
}{
	1 + \breve{\omega}_k
	\min_{m\in\mathcal{M}} \left\{ \frac{ \rho_{max}/\rho_s }{  \sum_{k=1}^K \gamma_{mk} } \right\}
}
\right\}
\end{align}
with $k\in\mathcal{K}\backslash\{1\}$. Once \eqref{nec. cond.} has been satisfied, the solution to (\underline{\textrm{P1}}) can be given by
\begin{itemize}
\item either
	\begin{align}
	\eta_{\left(\underline{\textrm{P1}}\right)}^{\star}
	= \min_{m\in\mathcal{M}}
	\left\{
	\frac{ \rho_{max}/\rho_s }{  \sum_{k=1}^K \gamma_{mk} }
	\right\}
	\end{align}
	for
	\begin{align}\label{thetaE: cond. 1}
	\theta_{\eav} \geq \varpi/\breve{\varpi}
	\end{align}
\item or
	\begin{align}
	\eta_{\left(\underline{\textrm{P1}}\right)}^{\star}
	= \min_{m\in\mathcal{M}}
	\left\{
	\frac{ \rho_{max}/\rho_s }{  \sum_{k=1}^K \gamma_{mk} } ,
	\frac{ \theta_{\eav} }{ \left( \varpi - \theta_{\eav} \breve{\varpi} \right) }
	\right\}
	\end{align}
	for
	\begin{align} \label{thetaE: cond. 2}
	\dfrac{ \varpi \max_{ k\in\mathcal{K}\backslash\{1\} } \left\{ \frac{\theta_k}{ \omega_k - \theta_k \breve{\omega}_k } \right\} }
	{ 1 + \breve{\varpi} \max_{ k\in\mathcal{K}\backslash\{1\} } \left\{ \frac{\theta_k}{ \omega_k - \theta_k \breve{\omega}_k } \right\} }
	\leq \theta_{\eav} < \varpi/\breve{\varpi} .
	\end{align}
\end{itemize}

\subsection{Closed-form solution to (\underline{\textrm{Q1}})}\label{subsec: under Q1}
As presented in the previous subsection, \eqref{nec. cond.} is necessary in order that (\underline{\textrm{Q1}}) can be solved. Then we can rewrite (\underline{\textrm{Q1}}) as
\begin{subequations}
	\begin{eqnarray}
	(\underline{\textrm{Q1}})
	~&\ds\max_{ \chi }
	~& l(\chi)
	\\
	~& \text{s.t.}
	~&  0 \leq \chi \leq \overline{\alpha}
	\end{eqnarray}
\end{subequations}
where
\begin{align}
\chi &\triangleq \eta - \underline{\alpha} ,
\nonumber \\
\overline{\alpha} &\triangleq
\min_{m\in\mathcal{M}} \left\{ \frac{ \rho_{max}/\rho_s }{  \sum_{k=1}^K \gamma_{mk} } \right\} - \underline{\alpha} ,
\nonumber \\
\underline{\alpha} &\triangleq \max_{ k\in\mathcal{K}\backslash\{1\} } \left\{ \frac{\theta_k}{ \omega_k - \theta_k \breve{\omega}_k } \right\}
\nonumber
\end{align}
and
\begin{align}
l(\chi) &=
	\tfrac{
		\chi \left( \omega_1+\breve{\omega}_1 \right) + \underline{\alpha} \left(\omega_1+\breve{\omega}_1\right) + 1
	}{
		\chi \breve{\omega}_1 + \underline{\alpha} \breve{\omega}_1 + 1
	}
	\tfrac{
		\chi \breve{\varpi} + \underline{\alpha} \breve{\varpi} + 1
	}{
		\chi \left( \varpi + \breve{\varpi} \right) + \underline{\alpha} \left( \varpi + \breve{\varpi} \right) + 1
	}.
	\nonumber
\end{align}
Introducing a new variable $\tau \geq 0$ and defining a Lagrangian function $\mathcal{L}\left(\chi, \tau \right) \triangleq l(\chi) - \tau (\chi - \overline{\alpha})$, we first consider two sub-cases:
\begin{itemize}
	\item For $\tau=0$, we solve $\frac{\partial l\left(\chi\right)}{\partial \chi} = 0$ to obtain two \emph{positive-real} critical points $\chi = \chi_1$ and $\chi = \chi_2$ (if possible).
	\item For $\tau>0$, we solve the system of two equations
\begin{align}
\left\{
\begin{array}{ll}
\dfrac{\partial \mathcal{L}\left(\chi, \tau \right)}{\partial \tau} = 0  \\
\dfrac{\partial \mathcal{L}\left(\chi, \tau \right)}{\partial \chi} = 0
\end{array}
\right.
\Leftrightarrow
\left\{
\begin{array}{ll}
\chi = \overline{\alpha} \\
\tau = \left. \dfrac{\partial l\left(\chi\right)}{\partial \chi} \right|_{ \chi = \overline{\alpha} }
\end{array}
\right.
\nonumber
\end{align}
to obtain another critical point $\chi = \overline{\alpha} \triangleq \chi_3$.
\end{itemize}
Then the optimal solution to $\left(\underline{\textrm{Q1}}\right)$ can be given by
\begin{align}
	\eta_{\left(\underline{\textrm{Q1}}\right)}^{\star}
	&= \underline{\alpha} + \argmax_{\chi \in\{ \chi_1, \chi_2, \chi_3 \} } l\left( \chi \right).
\end{align}

\subsection{Closed-form solution to (\underline{\textrm{R1}})}\label{subsec: under R1}
Similar to (\underline{\textrm{P1}}), we first need the condition \eqref{nec. cond.} with $k\in\{1,\ldots,K\}$ in order that (\underline{\textrm{R1}}) can be solved. Then we can attain the solution to (\underline{\textrm{R1}}), i.e.
\begin{align}
	\eta_{\left(\underline{\textrm{R1}}\right)}^{\star}
		= \max_{ k\in\mathcal{K} } \left\{ \frac{ \theta_k }{ \omega_k - \theta_k \breve{\omega}_k } \right\} ,
\end{align}
in the case that either \eqref{thetaE: cond. 1} or \eqref{thetaE: cond. 2} is satisfied.

\subsection{Closed-form solutions to (\underline{\textrm{S1}})}\label{subsec: under S1}
For (\underline{\textrm{S1}}), the condition \eqref{nec. cond.} (with $k\in\{2,\ldots,K\}$) is also required. The third constraint \eqref{special S1: c} is rewritten in the form $\breve{a} \eta^2 + \breve{b} \eta + \breve{c} \geq 0 $ with $\breve{a} = \breve{\varpi} \left( \omega_1 + \breve{\omega}_1 \right) - \phi \breve{\omega}_1 \left( \breve{\varpi} + \varpi \right) $, $\breve{b} = \omega_1 + \breve{\omega}_1 + \breve{\varpi} - \phi \left( \breve{\omega}_1 + \breve{\varpi} + \varpi \right) $ and $\breve{c} = 1 - \phi$. As such, there are two possibilities as follows:
\begin{itemize}
	\item If $\breve{a} > 0$, then \eqref{special S1: c} always holds for $\breve{b}^2 - 4 \breve{a} \breve{c} \leq 0$. In this case, the solution to (\underline{\textrm{S1}}) is given by
\begin{align}
\eta_{\left(\underline{\textrm{S1}}\right)}^{\star}
= \max_{ k\in\mathcal{K}\backslash\{1\} } \left\{ \frac{ \theta_k }{ \omega_k - \theta_k \breve{\omega}_k } \right\}.
\end{align}
	\item If $\breve{a} < 0$, then \eqref{special S1: c} holds for $\breve{b}^2 - 4 \breve{a} \breve{c} > 0$ and $\eta_1 \leq \eta \leq \eta_2$ given that $\eta_1$ and $\eta_2$ are the solutions to the quadratic equation $\breve{a} \eta^2 + \breve{b} \eta + \breve{c} = 0 $. In this case, (\underline{\textrm{S1}}) is infeasible if $\eta_2<0$; otherwise, the solution to (\underline{\textrm{S1}}) is given by
\begin{align}
\eta_{\left(\underline{\textrm{S1}}\right)}^{\star}
= \max_{ k\in\mathcal{K}\backslash\{1\} } \left\{ \frac{ \theta_k }{ \omega_k - \theta_k \breve{\omega}_k } , \eta_1\right\}.
\end{align}
\end{itemize}

\section{Numerical Results}\label{SEC: RESULT}
In this section, we evaluate the secure performance and make comparisons for different scenarios. More specifically, we measure the secure performance by calculating $R_{sec}$ (in nats/s/Hz) at
\begin{itemize}
	\item $\B{\Psi} = \B{\Psi}_{\left( \textbf{P1} \right)}^{\star}$ (the solution to $\left( \textbf{P1} \right)$);
	\item $\B{\Psi} = \B{\Psi}_{\left( \textbf{Q1} \right)}^{\star}$ (the solution to $\left( \textbf{Q1} \right)$);
	\item $\B{\Psi} = \B{\Psi}_{\left( \textbf{R1} \right)}^{\star}$ (the solution to $\left( \textbf{R1} \right)$);
	\item $\B{\Psi} = \B{\Psi}_{\left( \textbf{S1} \right)}^{\star}$ (the solution to $\left( \textbf{S1} \right)$);
	\item $\eta_{ \left(\underline{\textrm{P1}}\right) }^{\star}$ (the solution to $\left( \underline{\textrm{P1}} \right)$);
	\item $\eta_{ \left(\underline{\textrm{Q1}}\right) }^{\star}$ (the solution to $\left( \underline{\textrm{Q1}} \right)$);
	\item $\eta_{ \left(\underline{\textrm{R1}}\right) }^{\star}$ (the solution to $\left( \underline{\textrm{R1}} \right)$);
	\item $\eta_{ \left(\underline{\textrm{S1}}\right) }^{\star}$ (the solution to $\left( \underline{\textrm{S1}} \right)$).
\end{itemize}
For each case, the obtained value of $R_{sec}$ will be denoted by
$R_{sec} \left( \textbf{P1} \right)$,
$R_{sec} \left( \textbf{Q1} \right)$,
$R_{sec} \left( \textbf{R1} \right)$,
$R_{sec} \left( \textbf{S1} \right)$,
$R_{sec} \left( \underline{\textrm{P1}} \right)$,
$R_{sec} \left( \underline{\textrm{Q1}} \right)$,
$R_{sec} \left( \underline{\textrm{R1}} \right)$
and $R_{sec} \left( \underline{\textrm{S1}} \right)$, respectively.
Likewise, the notation
$P_{tot} \left( \textbf{R1} \right)$,
$P_{tot} \left( \textbf{S1} \right)$,
$P_{tot} \left( \underline{\textrm{R1}} \right)$,
and $P_{tot} \left( \underline{\textrm{S1}} \right)$
will stand for
\emph{``
	the total average transmit power of all APs at}
	$\B{\Psi} = \B{\Psi}_{\left( \textbf{R1} \right)}^{\star}$,
	$\B{\Psi} = \B{\Psi}_{\left( \textbf{S1} \right)}^{\star}$,
	$\eta = \eta_{\left( \underline{\textrm{R1}} \right)}^{\star}$,
	\emph{and} $\eta = \eta_{\left( \underline{\textrm{S1}} \right)}^{\star}$,
	\emph{respectively.''}

As for simulation parameters, we use the Hata-COST231 model (see \cite{Cell-free_AHien2017, Rappaport-Book-2002} and \cite{Chen-2006-VTC-HataModel}) to imitate the large scale fading coefficients, i.e.
\begin{align}
	\beta_{mk} &= 10^{ \left( \mathcal{S} + PL\left(d_{mk}\right) \right)/10 },
	\\
	\beta_{m\eav} &= 10^{ \left( \mathcal{S} + PL\left(d_{m\eav}\right) \right) /10 }
\end{align}
where $\mathcal{S} \sim \CN{0}{\sigma_{\mathcal{S}}^2}$ presents the shadowing fading effect with the standard deviation $\sigma_{\mathcal{S}} = 8$ dB and
\begin{align}\label{PL_mk}
	PL\left(d\right) &=
	\begin{cases}
		-139.4 - 35 \log_{10}(d) ~\textrm{if}~ d>0.05 \\
		-119.9 - 20 \log_{10}(d) ~\textrm{if}~ d \in (0.01, 0.05] \\
		-79.9 ~\textrm{if}~ d \leq 0.01
	\end{cases}
\end{align}
represents the path loss in dB with $d\equiv d_{mk}$ (or $d\equiv d_{m\eav}$) being the distance in km between the $m$th AP and user $k$ (or Eve).\footnote{Other presentations for $PL\left(d\right)$ are also available in literature. Herein, \eqref{PL_mk} is suggested for a practical scenario in which the carrier frequency is 1900 MHz, the heigh of each AP antenna is 20 m, the heigh of each user antenna (as well as that of Eve antenna) is 1.5 m and all nodes (APs, users and Eve) are randomly dispersed over a square of size $1\times 1$ km${}^{2}$ \cite[Eqs. (52) and (53)]{Cell-free_AHien2017}.} In addition, the maximum transmit power of each AP is $P_{max}=1$ W. Meanwhile, the average noise power (in W) is given by
\begin{align}
	N_0 = \textrm{bandwidth} \times k_B \times T_0 \times \textrm{noise~figure}
\end{align}
where $k_B = 1.38\times 10^{-23}$ (Joule/Kelvin) is the Boltzmann constant, and $T_0 = 290$ (Kelvin) is the noise temperature. In all simulation results, we suppose that the bandwidth is $20$ MHz and the noise figure is $9$ dB. Finally, other parameters will be mentioned whenever they are used.

\begin{figure}[t!]
   \centerline{\includegraphics[width=0.5\textwidth]{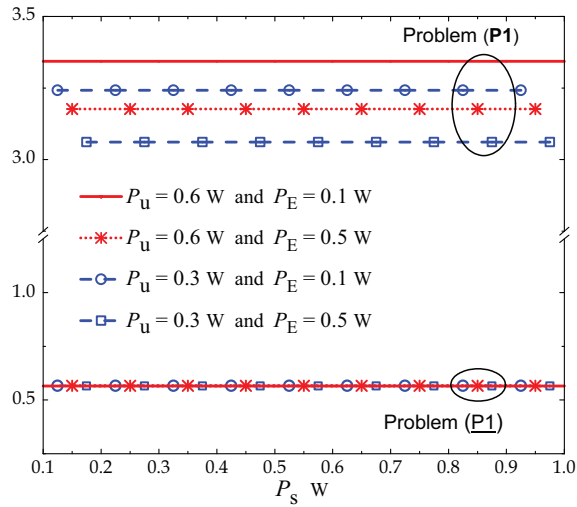}}
   \caption{Secrecy rate versus $P_s$ (the average transmit power for a signal $s_k$). Other parameters: the average transmit power of each user is $P_u = \{0.3, 0.6\}$ W, the average transmit power of Eve is $P_{\eav}= \{0.1, 0.5\}$ W, $M = 50$, $K = 8$, $T = 12$, $\theta_{\eav} = 10^{-4}$, and $\theta_k = 2\times 10^{-4}$ for $k=\{2,\ldots,K\}$.}
   \label{result 1}
\end{figure}


\begin{figure}[t!]
	\centerline{\includegraphics[width=0.5\textwidth]{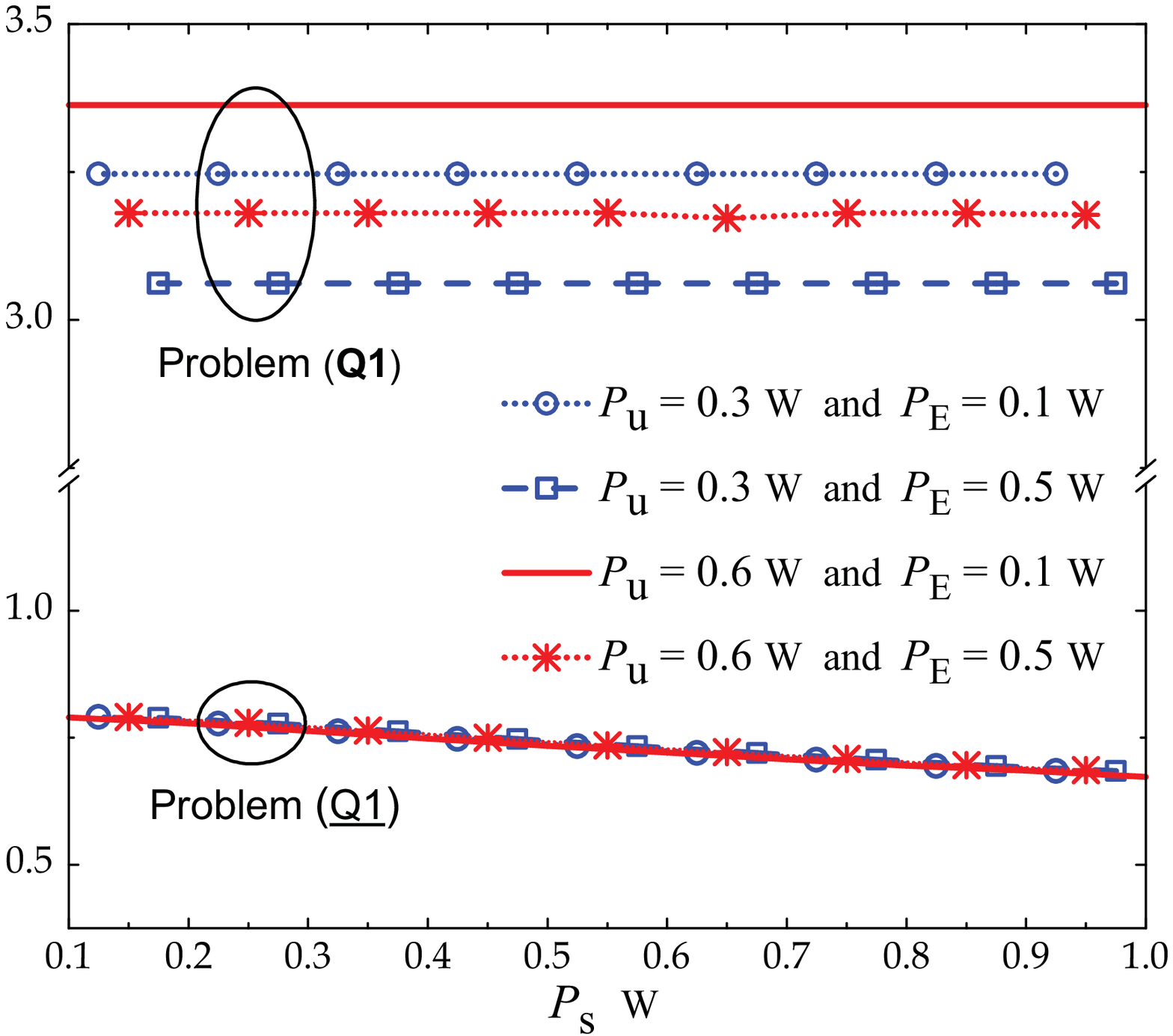}}
	\caption{Secrecy rate versus $P_s$ (the average transmit power for a signal $s_k$). Other parameters: the average transmit power of each user is $P_u = \{0.3, 0.6\}$ W, the average transmit power of Eve is $P_{\eav}= \{0.1, 0.5\}$ W, $M = 50$, $K = 8$, $T = 12$ and $\theta_k = 2\times 10^{-4}$ for $k=\{2,\ldots,K\}$.}
	\label{result 3}
\end{figure}

\begin{figure}[t!]
	\centerline{\includegraphics[width=0.5\textwidth]{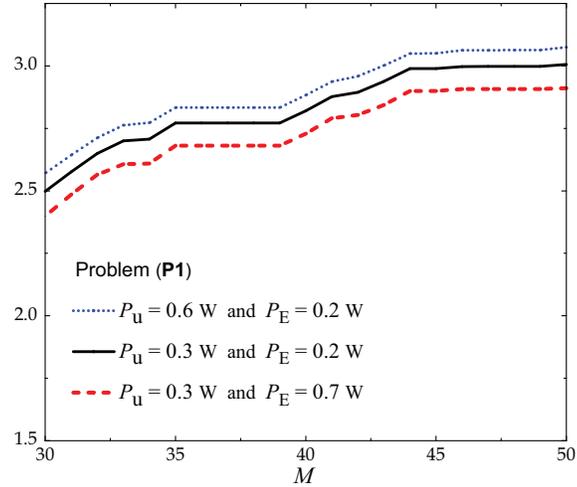}}
	\caption{Secrecy rate versus $M$. Other parameters: the average transmit power for a signal $s_k$ is $P_s = 0.8$ W, the average transmit power of each user is $P_u = \{0.3, 0.6\}$ W, the average transmit power of Eve is $P_{\eav}= \{0.2, 0.7\}$ W, $K=8$, $T = 12$, $\theta_k = 2\times 10^{-4}$ for $k=\{2,\ldots,K\}$ and $\theta_{\eav} = \theta_k / 50$.}
	\label{result 2b}
\end{figure}

\begin{figure}[t!]
	\centerline{\includegraphics[width=0.5\textwidth]{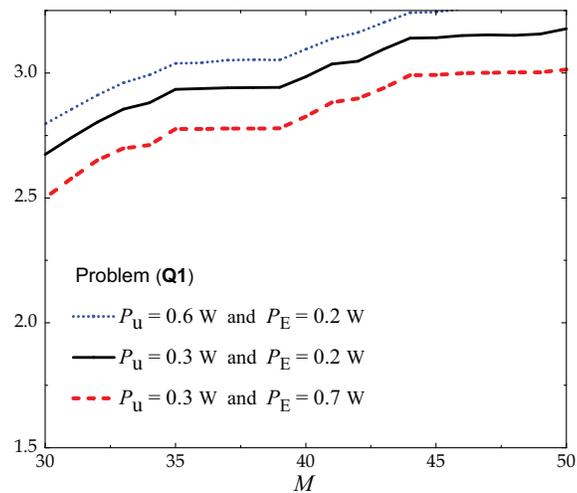}}
	\caption{Secrecy rate versus $M$. Other parameters: the average transmit power for a signal $s_k$ is $P_s = 0.8$ W, the average transmit power of each user is $P_u = \{0.3, 0.6\}$ W, the average transmit power of Eve is $P_{\eav}= \{ 0.2, 0.7\}$ W, $K = 8$, $T = 12$ and $\theta_k = 2\times 10^{-4}$ for $k=\{2,\ldots,K\}$.}
	\label{result 4}
\end{figure}

In Figure \ref{result 1}, we show the achievable secrecy rate (in nats/s/Hz) in 2 different cases: i) $\B{\Psi} = \B{\Psi}_{\left(\textbf{P1}\right)}^{\star}$ and ii) $\eta = \eta_{\left(\underline{\textrm{P1}}\right)}^{\star} $. For each case, $3$ different sub-cases of $\left(P_u, P_{\eav}\right)$ are considered. It is observed that $R_{sec} \left( \textbf{P1} \right)$ is significantly higher than $R_{sec} \left( \underline{\textrm{P1}} \right)$. In fact, the obtained values of $R_{sec} \left( \underline{\textrm{P1}} \right)$ fall within the interval $\left(0.55, 0.57\right)$ nats/s/Hz. In other words, having $\eta_{mk} = \eta^{\star}$ (for all $m$ and $k$) will lead to very poor performance in terms of security. Furthermore, the secure performance increases with $P_u$ and reduces with $P_{\eav}$ (the average transmit power of Eve).

Figure \ref{result 3} shows the achievable secrecy rate versus $P_s$ in two cases: i) $\B{\Psi} = \B{\Psi}_{\left(\textbf{Q1}\right)}^{\star}$ and ii) $\eta = \eta_{\left(\underline{\textrm{Q1}}\right)}^{\star} $. The secure performance in the first case is significantly higher than the second case. Moreover, the changes in the value of $R_{sec}\left(\underline{\textrm{Q1}}\right)$ are minor, i.e. $R_{sec}\left(\underline{\textrm{Q1}}\right)$ falls within $(0.67,0.79)$ nats/s/Hz. We also observe that $R_{sec} \left( \textbf{Q1} \right)$ is improved with increasing $P_u$ and is impaired with $P_{\eav}$. Meanwhile, $R_{sec}\left(\underline{\textrm{Q1}}\right)$ slightly decreases with $P_u$.

\begin{figure}[t!]
	\centerline{\includegraphics[width=0.5\textwidth]{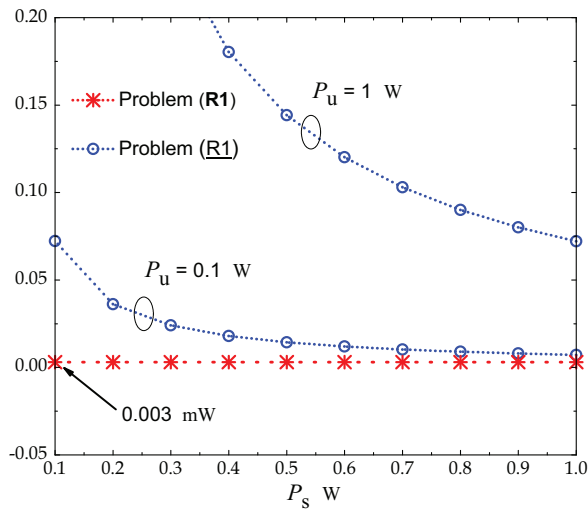}}
	\caption{Total power of all APs (in mW) versus $P_s$ (the average transmit power for a signal $s_k$). Other parameters: the average transmit power of each user is $P_u = \{0.1, 1\}$ W, the average transmit power of Eve is $P_{\eav}= 0.5$ W, $M = 50$, $K= 8$, $T = 12$, $\theta_1 = 0.1$, $\theta_k = 0.02$ for $k=\{2,\ldots,K\}$ and $\theta_{\eav} = \theta_1/50$.}
	\label{result 5}
\end{figure}

\begin{figure}[t!]
	\centerline{\includegraphics[width=0.5\textwidth]{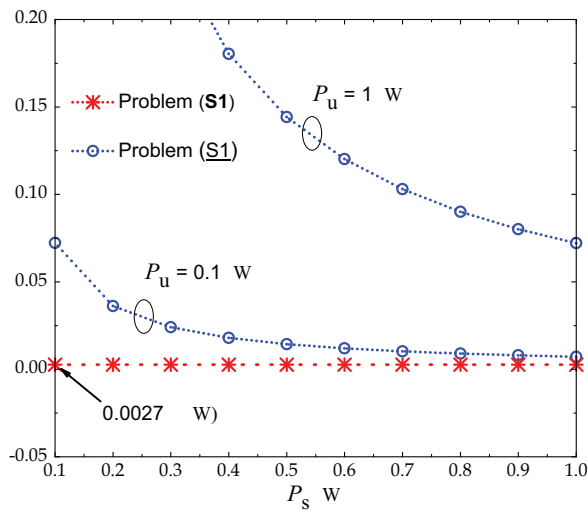}}
	\caption{Total power of all APs (in mW) versus $P_s$ (the average transmit power for a signal $s_k$). Other parameters: the average transmit power of each user is $P_u = \{0.1, 1\}$ W, the average transmit power of Eve is $P_{\eav}= 0.5$ W, $M = 50$, $K= 8$, $T = 12$, $\theta_k = 0.02$ for $k=\{2,\ldots,K\}$ and $\phi=1$.}
	\label{result 7}
\end{figure}

\begin{figure}[t!]
	\centerline{\includegraphics[width=0.5\textwidth]{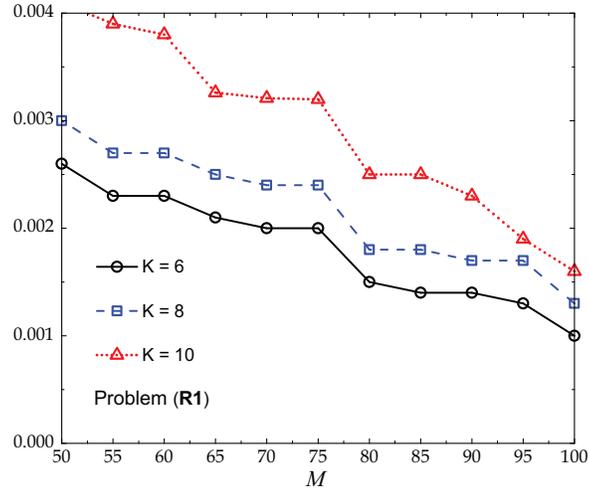}}
	\caption{Total power of all APs (in mW) versus $M$. Other parameters: the average transmit power for a signal $s_k$ is $P_s=0.7$ W, the average transmit power of each user is $P_u = 0.4$ W, the average transmit power of Eve is $P_{\eav}= 0.5$ W, $K= \{6,8,10\}$, $T = 12$, $\theta_1 = 0.1$, $\theta_k = 0.02$ for $k=\{2,\ldots,K\}$ and $\theta_{\eav} = \theta_1/50$.}
	\label{result 6}
\end{figure}

\begin{figure}[t!]
	\centerline{\includegraphics[width=0.5\textwidth]{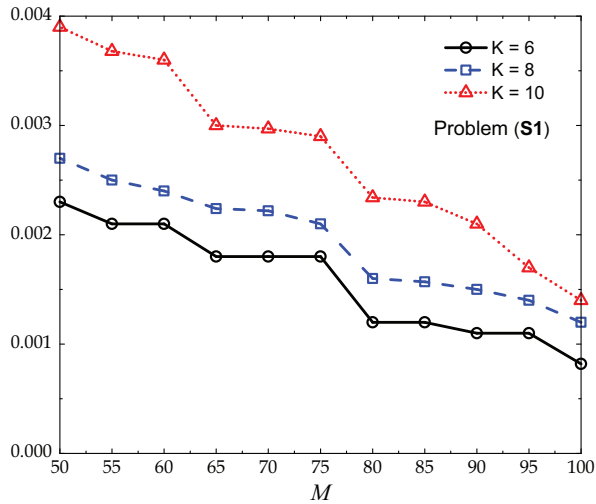}}
	\caption{Total power of all APs (in mW) versus $M$. Other parameters: the average transmit power for a signal $s_k$ is $P_s=0.7$ W, the average transmit power of each user is $P_u = 0.4$ W, the average transmit power of Eve is $P_{\eav}= 0.5$ W, $K = \{6,8,10\}$, $T = 12$, $\theta_k = 0.02$ for $k=\{2,\ldots,K\}$ and $\phi=1$.}
	\label{result 8}
\end{figure}

In Figures \ref{result 2b} and \ref{result 4}, the achievable secrecy rates $R_{sec}\left(\textbf{P1}\right)$ and $R_{sec}\left(\textbf{Q1}\right)$ are depicted as functions of $M$. We can see that both of them increase with $M$. It implies that the more service APs we have, the higher secure performance we gain. Finally, $R_{sec} \left(\textbf{P1}\right)$ as well as $R_{sec}\left(\textbf{Q1}\right)$ increases with $P_u$ and decreases with $P_{\eav}$. With the chosen parameters, $\left(\textbf{Q1}\right)$ appears better than $\left(\textbf{P1}\right)$ in terms of secrecy rate. Overall, $P_{\eav}$ represents the strength of an actively eavesdropping attack; thus, we can observe that the secure performance is degraded when $P_{\eav}$ grows as shown in Figures 2-5.

Figure \ref{result 5} shows that $P_{tot} \left( \underline{\textrm{R1}} \right)$ is much higher than $P_{tot} \left( \textbf{R1} \right)$ which is around $0.003$ mW with every $P_s$. It means that the solution $\B{\Psi}_{\left(\textbf{R1}\right)}^{\star}$ is much better than the solution $\eta_{ \left(\underline{\textrm{P1}}\right) }^{\star} $ in terms of energy, because the APs do not have to consume too much energy to meet security requirements. Besides, the figure also shows that $P_{tot} \left( \underline{\textrm{R1}} \right)$ inversely decreases with $P_s$ and is lowest at $P_s=P_{max}$. Finally, we observe that when $P_s$ changes, $R_{sec} \left( \textbf{R1} \right) \approx 0.0953$ nats/s/Hz remains almost constant; meanwhile, $R_{sec} \left( \underline{\textrm{R1}} \right) \approx 0.5386$ nats/s/Hz with $P_u=0.1$ W and $0.5091$ nats/s/Hz with $P_u=1$ W.

Figure \ref{result 7} shows that $P_{tot} \left( \underline{\textrm{S1}} \right)$ is much higher than $P_{tot} \left( \textbf{S1} \right)$ which is around $0.0027$ mW at each considered value of $P_s$. This result also reveals that $\left(\textbf{R1}\right)$ is the better program in terms of energy, because there is really less energy required for security. Besides, the figure also shows that $P_{tot} \left( \underline{\textrm{S1}} \right)$ inversely decreases with $P_s$ and is lowest at $P_s=P_{max}$. Finally, we record that $R_{sec} \left( \textbf{S1} \right) \approx 0^{+}$ nats/s/Hz when $P_s$ changes. In contrast, $R_{sec} \left( \underline{\textrm{S1}} \right) \approx 0.4619$ nats/s/Hz with $P_u=0.1$ W and $0.4521$ nats/s/Hz with $P_u=1$W.

Figure \ref{result 6} depicts $P_{tot} \left( \textbf{R1} \right)$ as a function of $M$. With 3 different values of $K$, we observe that the total power consumption reduces with $M$ but increases with $K$. We can see that $\left( \textbf{R1} \right)$ can be solved with many different values of $(M,K)$. Among them, the best choice is to choose $M$ as large as possible while $K$ should be as small as possible. For example, the system with $(M,K)=(70,6)$ will require less power consumption (at APs) than the system with $(M,K)=(50,10)$, while the security constraints remain guaranteed.

Figure \ref{result 8} depicts $P_{tot} \left( \textbf{S1} \right)$ as a function of $M$. Our observation of this figure is similar to Figure \ref{result 6}. We should choose $M$ as large as possible and $K$ as small as possible in order to attain the best performance (as long as the security constraints are satisfied). When $M$ is large enough, the total power consumption is nearly zero and yet, the secrecy rate is also around zero (with the chosen parameters).

In comparison between Figure \ref{result 6} and Figure \ref{result 8}, one can find the two differences: i) the presence of $\theta_{\eav}$ and the absence of $\phi$ in Figure \ref{result 6}; and ii) the absence of $\theta_{\eav}$ and the absence of $\phi$ in Figure \ref{result 8}. It is because of the fact that $\left( \textbf{R1} \right)$ and $\left( \textbf{S1} \right)$ have different security constraints. With the setup parameters, $\left( \textbf{S1} \right)$ offers better performance than $\left( \textbf{R1} \right)$ because the required power consumption is lower (i.e., the curves in Figure \ref{result 8} is slightly lower than those in Figure \ref{result 6}).

\section{Conclusions}\label{SEC: CON}
In this paper, we have considered a cell-free MIMO network in the presence of an active eavesdropper. We have suggested maximization problems to maximize the achievable secrecy rate subject to quality-of-service constraints. Also, minimization problems have been provided to minimize power consumption as long as security requirements are still guaranteed. In finding the optimal values of the power control coefficients $\{\eta_{mk}\}_{m,k}$, we have considered two different cases: i) $\eta_{mk}$ changes with $m$ and $k$; and ii) $\eta_{mk} = \eta$ for all $m$ and $k$. Through numerical results, we have found that the case of $\eta_{mk} = \eta$ will lead to far worse performance than the other case. Based on numerical results and intuitive observations, a trade-off problem between secrecy rate and energy consumption may be considered for cell-free networks in the future. Besides, preventing Eve's intrusion into the pilot training will be also worth considering.

%

\appendix
\subsection{A Simple Method to Identify Abnormality in Pilot Training} \label{Appendix: method}
As presented in Subsection II.A, the $m$th AP receives the array of signals $\{ y_{km} \}_{k=1}^K$ after calculating the Hermitian inner product between $\B{y}_{p,m}$ and $\B{p}_k$. Then all APs (through the CPU) exchange information and make a calculation of
\begin{align}
\mathcal{Y} &\triangleq \sum_{m=1}^{M} \EX{ | y_{1m} |^2 }
\nonumber
\end{align}
to check if Eve tries to overhear the signal transmitted from APs to user 1. If $\mathcal{H}_0$ denotes the hypothesis that there is no active eavesdropping and $\mathcal{H}_1$ denotes the opposite, then two possibly obtained values of $\mathcal{Y}$ are
\begin{align}
	\mathcal{Y}|_{\mathcal{H}_0} &= T \rho_u \sum_{m=1}^{M} \beta_{m1} + M,
	\nonumber \\
	\mathcal{Y}|_{\mathcal{H}_1} &= T \rho_u \sum_{m=1}^{M}  \beta_{m1} + T \rho_{\eav} \sum_{m=1}^{M} \beta_{m\eav} + M.
	\nonumber
\end{align}
It is clear that $\mathcal{Y} |_{\mathcal{H}_1} > \mathcal{Y} |_{\mathcal{H}_0}$ always holds for $\rho_{\eav}>0$. Therefore, APs simply compare $\mathcal{Y}$ with $\mathcal{Y}|_{\mathcal{H}_0}$ to make the decision, i.e.
\begin{itemize}
	\item $\mathcal{Y} = \mathcal{Y} |_{\mathcal{H}_0} \Leftrightarrow$ No active eavesdropping.
	\item $\mathcal{Y} > \mathcal{Y} |_{\mathcal{H}_0} \Leftrightarrow$ Eve is seeking to attack the system.
\end{itemize}
Note that $\mathcal{Y} |_{\mathcal{H}_0}$ is a \emph{known} value and is referred to as the only threshold (which APs need) to check any abnormality in pilot training related to the pilot $\B{p}_1$.

In fact, the above-mentioned detection method can be performed without knowing the value of $\rho_{\eav}$. However, $\rho_{\eav}$ can also be predicted by
$$
\rho_{\eav} = \frac{ \mathcal{Y} - \mathcal{Y} |_{\mathcal{H}_0} }{ T \sum_{m=1}^{M} \beta_{m\eav} }
$$
in the case that active eavesdropping occurs.

\subsection{Explicit expression for $\snrE{}$} \label{Appendix: theo: 1}
We first calculate
\begin{align}\label{BU_E}
  &\EX{|\mathrm{BU}_{\eav, 1}|^2}
    = \rho_s \sum_{m=1}^{M} \eta_{m1} \EX{ \left| g_{m\eav} \hat{g}_{m1}^*  \right|^2 }
    \nonumber \\
    &\mathop=\limits^{(a)}
    \rho_s \sum_{m=1}^{M} \eta_{m1}
        \EX{ \left| \left(e_{m\eav} + \hat{g}_{m\eav}\right) \hat{g}_{m1}^* \right|^2 }
    \nonumber \\
    &=
    \rho_s \sum_{m=1}^{M} \eta_{m1}
    \left( \EX{ \left| \hat{g}_{m\eav} \hat{g}_{m1}^* \right|^2 } + \EX{ \left| e_{m\eav} \right|^2 \left| \hat{g}_{m1}^* \right|^2 }
    \right)
    \nonumber \\
    &\mathop=\limits^{(b)}
    	\rho_s \sum_{m=1}^{M} \eta_{m1}
        \left[
            2 \alpha_m \gamma_{m1}^2
            + \left(\beta_{m\eav}-\gamma_{m\eav}\right) \gamma_{m1}
        \right]
    \nonumber \\
    &\mathop=\limits^{(c)}
    \rho_s \sum_{m=1}^{M} \eta_{m1}
	\left( \alpha_m \gamma_{m1}^2 + \beta_{m\eav} \gamma_{m1}\right)
	\nonumber \\
	&=
    \rho_s \sum_{m=1}^{M} \eta_{m1}
	\left[
	\left(
		\frac{ \rho_{\eav}\beta_{m\eav}^2 }{\rho_u\beta_{m1}^2}
	\right)
	\gamma_{m1}^2
	+ \beta_{m\eav} \gamma_{m1}
	\right]
\end{align}
where $(a)$ is obtained by substituting $g_{m\eav} = e_{m\eav} + \hat{g}_{m\eav}$ with $e_{m\eav} \triangleq g_{m\eav} - \hat{g}_{m\eav}$ being the channel estimation error for the link between the $m$th AP and Eve. In deriving $(a)$, we also use the fact that $e_{m\eav} \sim \CN{0}{\beta_{m\eav}-\gamma_{m\eav}}$ is independent of $\hat{g}_{m\eav}$. The equality $(b)$ is obtained by using \eqref{eq: pro 2}. Meanwhile, $(c)$ results from the substitution of $\gamma_{m\eav} = \alpha_m \gamma_{m1}$.

Similarly, for $k'\neq 1$, we calculate
\begin{align}\label{UI_E}
  &\EX{ \left| \mathrm{UI}_{\eav, k'} \right|^2}
    = \rho_s \EX{ \left| \sum_{m=1}^{M}\sqrt{\eta_{mk'}}g_{m\eav}\hat{g}_{mk'}^* \right|^2}
    \nonumber \\
    &= \rho_s \sum_{m=1}^{M} \eta_{mk'}
    \left( \EX{\left| \hat{g}_{m\eav}\hat{g}_{mk'}^* \right|^2}
            + \EX{\left| e_{m\eav}\hat{g}_{mk'}^* \right|^2} \right)
    \nonumber \\
    &\mathop=\limits^{(a)} \rho_s \sum_{m=1}^{M} \eta_{mk'}
    \left( \alpha_m
            \EX{\left| \hat{g}_{m1} \hat{g}_{mk'}^* \right|^2}
            + \EX{\left| e_{m\eav} \right|^2 \left| \hat{g}_{mk'}^* \right|^2} \right)
    \nonumber \\
    &\mathop=\limits^{(b)} \rho_s \sum_{m=1}^{M} \eta_{mk'}
    \left[ \alpha_m
        \gamma_{m1} \gamma_{mk'} +  \left(\beta_{m\eav}-\alpha_m \gamma_{m1}\right) \gamma_{mk'} \right]
    \nonumber \\
    &= \rho_s \sum_{m=1}^{M} \eta_{mk'} \beta_{m\eav} \gamma_{mk'}
\end{align}
where $(a)$ is obtained by using \eqref{phase 1: MMSE estimate 2} and $(b)$ results from the substitution of \eqref{eq: pro 1}.

Finally, substituting \eqref{BU_E} and \eqref{UI_E} into \eqref{phase 2: SNRe_1} yields \eqref{phase 2: SNRe_2}.

%

\bibliographystyle{IEEEtran}


\balance
\end{document}